\documentclass[10pt,conference,letterpaper]{style/IEEEtran}

\usepackage{epsfig}
\usepackage{amssymb}
\usepackage{amsmath}
\usepackage{color}
\usepackage{times}
\usepackage{arydshln}
\usepackage{url}
\usepackage{algorithm}
\usepackage{algorithmic}

\usepackage{flushend}
\usepackage{float}
\usepackage{subcaption}
\usepackage{todonotes}
\usepackage{balance}

\title{\huge Assisting Service Providers In Peer to peer Marketplaces:\\ Maximizing Gain Over Flexible Attributes}

\author{
{Abolfazl Asudeh{\small $~^{\dag}$}, Azade Nazi{\small $~^{\ddag}$}, Nick Koudas{\small $~^{\dag\dag}$}, Gautam Das{\small $~^{\dag}$} }%
\vspace{1.6mm}\\
\fontsize{10}{10}\selectfont\itshape
$^{\dag}$\,University of Texas at Arlington, $^\ddag$\,Microsoft Research, $^{\dag\dag}$\,University of Toronto\\
\fontsize{9}{9}\selectfont\ttfamily\upshape
$^{\dag}$\{ab.asudeh@mavs,~gdas@cse\}.uta.edu, $^\ddag$aznazi@microsoft.com, $^{\dag\dag}$koudas@cs.toronto.edu
}
\date{}

\newtheorem{theorem}{Theorem} 

\newtheorem{definition}{Definition}
\usepackage[font={small,bf}]{caption} 
\renewcommand{\baselinestretch}{0.99}

\begin{document}
\maketitle

\begin{abstract}
Peer to peer marketplaces enable transactional exchange of services directly between people. In such platforms, those providing a service are faced with various choices. For example in travel peer to peer marketplaces, although some amenities (attributes) in a property are fixed, others are relatively {\em flexible} and can be provided without significant effort. Providing an attribute is usually associated with a cost. Naturally, different sets of attributes may have a different ``gains'' (monetary or otherwise) for a service provider. Consequently, given a limited budget, deciding which attributes to offer is challenging. 

In this paper, we formally introduce and define the problem of {\em Gain Maximization over Flexible Attributes} (GMFA). We first prove that the problem is NP-hard and show that there is no approximate algorithm with a constant approximate ratio for it, unless there is one for the quadratic knapsack problem.
Instead of limiting the application to a specific gain function, we then provide a practically efficient exact algorithm to the GMFA problem 
that can handle any gain function as long as it belongs to the general class of monotonic functions.
As the next part of our contribution, we introduce the notion of {\em frequent-item based count} (FBC), which utilizes nothing but the existing tuples in the database to define the notion of gain in the absence of extra information.
We present the results of a comprehensive experimental evaluation of the proposed techniques on real dataset from AirBnB
and a case study that confirm the efficiency and practicality of our proposal.
\end{abstract}

\section{Introduction}
\vspace{0.05in}
\noindent
Peer to peer marketplaces enable both ``obtaining'' and ``providing'' in a temporary or permanent fashion valuable services through direct interaction between people~\cite{ertz2016collaborative}.
Travel peer to peer marketplaces such as {\it AirBnB}, {\it HouseTrip}, {\it HomeAway}, and {\it Vayable}\footnote{\small{airbnb.com; housetrip.com; homeaway.com; vayable.com.}}, work and service peer to peer marketplaces such as {\it UpWork}, {\it FreeLancer}, {\it PivotDesk}, {\it ShareDesk}, and {\it Breather}\footnote{\small{upwork.com; freelancer.com; pivotdesk.com; sharedesk.net; breather.com}}, 
car sharing marketplaces such as {\it BlaBlaCar}\footnote{\small{blablacar.com}}, education peer to peer marketplaces such as {\it PopExpert}\footnote{\small{popexpert.com}}, and pet peer to peer marketplaces such as {\it DogVacay}\footnote{\small{dogvacay.com}}
are a few examples of such marketplaces.
In travel peer to peer marketplaces, for example, the service caters to accommodation rental;
hosts are those providing the service ({\em service providers}), and guests, who are looking for temporary rentals, are receiving service
({\em service receivers}).
Hosts list properties, along with a set of amenities for each, while guests utilize the search interface to identify 
suitable properties to rent. Figure~\ref{fig:sampleset} presents a sample set of rental accommodations.
Each row corresponds to a property and each column represents an amenity. For instance, the first property offers 
\texttt{Breakfast}, \texttt{TV}, and \texttt{Internet} as amenities but does not offer \texttt{Washer}.

\begin{figure}[!t]
\centering
\begin{small}
\begin{tabular}{|l|c|c|c|c|}
    \hline 
         {\bf ID}    & {\bf Breakfast} & {\bf TV} & {\bf Internet} & {\bf Washer}    \\ \hline 
     Accom. 1  & 1  & 1			 & 1       & 0       \\ \hline
     Accom. 2  & 1  & 1			 & 1       & 1       \\ \hline
     Accom. 3  & 0  & 1			 & 1       & 0       \\ \hline
     Accom. 4  & 1  & 1			 & 1       & 0       \\ \hline
     Accom. 5  & 0  & 1			 & 1       & 1       \\ \hline
     Accom. 6  & 1  & 0			 & 1       & 0       \\ \hline
     Accom. 7  & 1  & 0			 & 0       & 0       \\ \hline
     Accom. 8  & 1  & 1			 & 0       & 1       \\ \hline
     Accom. 9  & 0  & 1			 & 1       & 1       \\ \hline
     Accom. 10 & 1  & 0			 & 0       & 1       \\ \hline
\end{tabular}
\end{small}
\caption{A sample set of rental accomodations}\label{fig:sampleset}
\end{figure} 

Although sizeable effort has been devoted to design user-friendly search tools assisting service receivers in the search process, 
little effort has been recorded to date to build tools to assist service providers.
Consider for example a host in a travel peer to peer marketplace; while listing a property in the service for (temporary) rent, the host is faced with various choices.
Although some amenities in the property are relatively {\em fixed}, such as number of rooms for rent, or existence
of an elevator, others are relatively {\em flexible}; for example offering \texttt{Breakfast} or \texttt{TV} as an amenity. Flexible amenities
can be added without a significant effort.
Although amenities make sense in the context of travel peer to peer marketplaces (as part of the standard terminology used in the service), 
for a general peer to peer marketplace we use the term {\em attribute} and refer to the subsequent choice of attributes as
{\em flexible attributes}.

Service providers participate in the service with specified objectives; for instance hosts may want to increase overall occupancy and/or optimize their anticipated revenue.
Since there is a {\em cost} (e.g., monetary base cost to the host to offer internet) associated with each flexible attribute, it is challenging for service providers to choose the set of flexible
attributes to offer given some budget limitations (constraints). An informed choice of attributes to offer should maximize the objectives of the service provider in each case subject to any constraints. Objectives may vary by application; for example an objective could be maximize the number of times a listing appears on search results, the position in the search result ranking or other. This necessitates the existence of functions that relate flexible attributes to such objectives in order to aid the service provider's decision. We refer to the service provider's objectives in a generic sense as {\em gain} and to the functions that relate attributes to gain as {\em gain functions} in what follows.

In this paper, we aim to assist service providers in peer to peer marketplaces by suggesting those flexible attributes which maximize their gain. Given a service with known flexible attributes and budget limitation, our objective is to identify a set of attributes to suggest to service providers in order to maximize the gain. We refer to this problem as {\em Gain Maximization over Flexible Attributes} ({\bf GMFA}). Since the target applications involve mainly ordinal attributes, in this paper, we focus our attention on ordinal attributes and we assume that numeric attributes (if any) are suitably discretized. Without loss of generality, we first design our algorithms for binary attributes, and provide the extension to ordinal attributes in \S~\ref{subsec:disc-categorical}.

Our contribution in this paper is twofold. First, we formally define the general problem of {\em Gain Maximization over Flexible Attributes} ({\bf GMFA}) in peer to peer marketplaces and, as our main contribution, propose a general solution which is applicable to a general class of gain functions.
Second, without making any assumption on the existence extra information other than the dataset itself, we introduce the notion of {\em frequent-item based count} as a simple yet compelling gain function in the absence of other sources of information.
As our first contribution, using a reduction from the quadratic knapsack problem~\cite{quadraticknapsack, garyjohnson}, we prove that (i) the general GMFA is NP-hard, 
and (ii) that there is no approximate algorithm with a fixed ratio for GMFA unless there is one for quadratic knapsack problem.
We provide a (practically) efficient exact algorithm to the GMFA problem for a general class of monotonic gain functions\footnote{\small{Monotonicity of the gain function simply means that adding a new attribute does not reduce the gain.}}.
This generic proposal is due to the fact that gain function design is application specific and depends on the available information.
Thus, instead of limiting the solution to a specific application, the proposed algorithm gives the freedom to easily apply any arbitrary gain function into it.
In other words, it works for {\em any arbitrary monotonic gain function} no matter how and based on what data it is designed.
In a rational setting in which
attributes on offer add value, we expect that all gain functions will be monotonic.
More specifically, given any user defined monotonic gain function, we first propose an algorithm called {\bf I-GMFA} (Improved GMFA) that exploits the properties of the monotonic function to suggest efficient ways to explore the solution space. Next, we introduce several techniques to speed up this algorithm, both theoretically and practically, changing the way the traversal of the search space is performed. 
To do so, we propose the {\bf G-GMFA} (General GMFA) Algorithm which transforms the underlying problem structure from a lattice to a tree, preorderes the the attributes, and amortizes the computation cost over different nodes during the traversal. 

The next part of our contribution, focuses on the gain function design.
It is evident that gain functions could vary depending on the underlying objectives and extra information
such as a weighting of attributes based on some criteria (e.g., importantance), that can be
naturally incorporated in our framework without changes to the algorithm.
The gain function design,  as discussed in Appendix~\ref{subsec:dis-userpref}, is application specific and may vary upon on the availability of information such as query logs or reviews; thus, rather than assuming the existence of any specific extra information, 
we, alternatively, introduce the notion of {\em frequent-item based count (FBC)} that utilizes nothing but the existing tuples in the database to define the notion of gain for the absence of extra information.
Therefore, even the applications with limited access to the data such as a third party service for assisting the service providers that may only have access to the dataset tuples can utilize G-GMFA while applying FBC inside it.
The motivation behind the definition of FBC is that (rational) service providers provide attributes based on demand.
For example, in Figure~\ref{fig:sampleset} the existence of \texttt{TV} and \texttt{Internet} together in more than half of the rows, indicates the demand for this combination of amenities. Also, as shown in the real case study provided in \S~\ref{subsec:casestudy}, popularity of \texttt{Breakfast} in the rentals located in Paris indicates the demand for this amenity there.
Since counting the number of frequent itemsets is \#P-complete~\cite{gunopulos2003discovering}, computing the FBC is challenging.
In contrast with a simple algorithm that is an adaptation of Apriori~\cite{apriori} algorithm, we propose a practical output-sensitive algorithm for computing FBC that runs in the time linear in its output value. The algorithm uses an innovative approach that avoids iterating over the frequent attribute combinations by partitioning them into disjoint sets and calculating the FBC as the summation of their cardinalities.

In summary, we make the following contributions in this paper.
\begin{itemize}
\item We introduce the notion of flexible attributes and the novel problem of gain maximization over flexible attributes (GMFA) in peer to peer marketplaces.
\item We prove that the general GMFA problem is NP-hard and we prove the difficulty of designing an approximate algorithm.
\item For the general GMFA problem, we propose an algorithm called {\bf I-GMFA} (Improved GMFA) that exploits the properties of the monotonic function to suggest efficient ways to explore the solution space.
\item We propose the {\bf G-GMFA} (General GMFA) algorithm which transforms the underlying problem structure from a lattice to a tree, preorders the attributes, and amortizes the computation cost over nodes during the traversal. Given the application specific nature of the gain function design, {\bf G-GMFA} is designed such that any arbitrary monotonic gain function can simply get plugged into it.
\item While not promoting any specific gain function, without any assumption on the existence of extra information other than the dataset itself, we propose {\em frequent-item based count (FBC)} as a simple yet compelling gain function in the absence of other sources of data.
\item In contrast with the simple Apriori-based algorithm, we propose and present the algorithm {\bf FBC} to efficiently assess gain and demonstrate its practical significance.
\item We present the results of a comprehensive performance study on real dataset from AirBnB to evaluate the proposed algorithms. Also, in a real case study, we to illustrate the practicality of the approaches.
\end{itemize}

This paper is organized as follows. 
\S~\ref{sec:preliminaries} provides formal definitions and introduces notation stating formally the problem we focus and its associated complexity.
We propose the exact algorithm for the general class of monotonic gain functions in \S~\ref{sec:exact}. In \S~\ref{sec:gain}, we study the gain function design and propose 
an alternative gain function for the absence of user preferences. The experiment results are provided in \S~\ref{sec:exp}, related work is discussed in \S~\ref{sec:related}, and the paper is concluded in \S~\ref{sec:conclusion}.

\section{Preliminaries}
\label{sec:preliminaries}
\noindent{\bf Dataset Model:}
We model the entities under consideration in a peer to peer marketplace as a dataset $\mathcal{D}$ with $n$ tuples and $m$ attributes $\mathcal{A}=\{A_1,\dots , A_m\}$. For a tuple $t \in \mathcal{D}$, we use $t[A_i]$ to denote the value of the attribute $A_i $ in $t$. Figure~\ref{fig:sampleset} presents a sample set of rental accommodations with $10$ tuples and $4$ attributes. Each row corresponds to a tuple (property) and each column represents an attribute. For example, the first property offers \texttt{Breakfast}, \texttt{TV}, and \texttt{Internet} as amenities but does not offer \texttt{Washer}. 
Note that, since the target applications involve mainly ordinal attributes, we focus our attention on such attributes and we assume that numeric attributes (if any) are suitably discretized.
Without loss of generality, throughout the paper, we consider the attributes to be binary and defer the extension of algorithms to ordinal attributes in \S~\ref{subsec:disc-categorical}.
We use $\mathcal{A}_t$ to refer to the set of attributes for which $t[A_i]$ is non zero; i.e. $\mathcal{A}_t=\{A_i\in\mathcal{A}~|~t[A_i]\neq 0\}$, and the size of $\mathcal{A}_t$ is $k_t$.

\noindent{\bf Query Model:}
Given the dataset $\mathcal{D}$ and set of binary attributes $\mathcal{A}'\subseteq \mathcal{A}$ , the query $Q(\mathcal{A}', \mathcal{D})$ returns the set of tuples in $\mathcal{D}$ where contain $\mathcal{A}'$ as their attributes; formally:

\begin{align}\label{eq:query}
Q(\mathcal{A}', \mathcal{D}) = \{t\in \mathcal{D} | \mathcal{A}'\subseteq \mathcal{A}_t \}
\end{align}

Similarly, the query model for the ordinal attributes is as following: given the dataset $\mathcal{D}$, the set of ordinal attributes $\mathcal{A}'\subseteq \mathcal{A}$, and values $\mathcal{V}$ where $V_{i}\in \mathcal{V}$ is a value in the domain of $A_i\in\mathcal{A}'$,  $Q(\mathcal{A}', \mathcal{V}, \mathcal{D})$ returns the tuples in $\mathcal{D}$ that for attribute $A_i\in\mathcal{A}'$, $V_i\leq t[A_i]$.

\noindent{\bf Flexible Attribute Model:}
In this paper, we assume an underlying \emph{cost}\footnote{\small{Depending on the application it may represent a monetary value.}} associated with each attribute $A_i$, i.e., a flexible attribute $A_i$ can be added to a tuple $t$ by incurring $cost[A_i]$. For example, the costs of providing attributes \texttt{Breakfast}, \texttt{TV}, \texttt{Internet}, and \texttt{Washer}, in Figure~\ref{fig:sampleset}, on an annual basis, are $cost=[1000, 300,\\ 250, 700]$. For the ordinal attributes, $cost( A_i,V_1,V_2)$ represents the cost of changing the value of $A_i$ from $V_1$ to $V_2$. Our approach places no restrictions on the number of
flexible attributes in $\mathcal{A}$. For the ease of explanation, in the rest of paper we assume all the attributes in $\mathcal{A}$ are flexible.

\noindent
We also assume the existence of a gain function $gain(.)$, that given the dataset $\mathcal{D}$, for a given attribute combination $\mathcal{A}_i\subseteq \mathcal{A}$, provides a score showing how desirable $\mathcal{A}_i$ is. For example in a travel peer to peer marketplace, given a set of $m$ amenities, such a function could quantify the anticipated gain 
(e.g., visibility) for a host if a subset of these amenities are provided by the host on a certain property. 

Table~\ref{table:notations} presents a summary of the notation used in this paper. We will provide the additional notations for \S~\ref{sec:gain} at Table~\ref{table:notations2}.
Next, we formally define the general {\em Gain Maximization over Flexible Attributes} ({\bf GMFA}) in peer to peer marketplaces. 
\begin{table}[!t]
\center
\footnotesize
\vspace{5mm}
\caption{Table of Notations}
\begin{tabular}{|@{}c@{}|@{}c@{}|}
\hline
{\bf Notation}& {\bf Meaning}\\ \hline
$\mathcal{D}$& The dataset\\ \hline
$\mathcal{A}$& The set of the attributes in database $\mathcal{D}$\\ \hline 
$m$& The size of $\mathcal{A}$\\ \hline
$n$& The number of tuples in database $\mathcal{D}$ \\ \hline
$t[A_i]$& The value of attribute $A_i$ in tuple $t$\\ \hline
$\mathcal{A}_t$& The set of non-zero attributes in tuple $t$\\ \hline
$cost[A_i]$& The cost to change the binary attribute $A_i$ to $1$ \\ \hline
$B$& The budget \\ \hline
$gain(.)$& The gain function \\ \hline
$\mathcal{L}_{\mathcal{A}_i}$ & The lattice of attribute combinations $\mathcal{A}_i$ \\ \hline 
$V(\mathcal{L}_{\mathcal{A}_i})$ &  The set of nodes in  $\mathcal{L}_{\mathcal{A}_i}$\\ \hline
$\mathcal{B}(v_i)$ &  The bit representative of the node $v_i$\\ \hline
$v(\mathcal{A}_i)$ &  The node with attribute combination $\mathcal{A}_i$\\ \hline
$v(\mathcal{B}_i)$ &  The node with the bit representative $\mathcal{B}_i$\\ \hline
$\ell(v_i)$ &  The level of the node $v_i$\\ \hline
$cost(v_i)$ &  The cost associated with the node $v_i$\\ \hline
parents($v_j$,$\mathcal{L}_{\mathcal{A}_i}$) &  The parents of the node $v_j$ in $\mathcal{L}_{\mathcal{A}_i}$\\ \hline
$\rho(\mathcal{B}(v_i))$ &  The index of the right-most zero in $\mathcal{B}(v_i)$\\ \hline
parent$_T(v_i)$ &  The parent of $v_i$ in the tree data structure\\ \hline 
\end{tabular}
\label{table:notations}
\end{table}

\subsection{General Problem Definition}
\label{sec:generalProb}
We define the general problem of {\em Gain Maximization over Flexible Attributes} ({\bf GMFA}) in peer to peer marketplaces as a constrained optimization problem. The general problem is agnostic to the choice of the gain function. Given $gain(.)$, a service provider with a certain budget $B$ strives to maximize
$gain(.)$ by considering the addition of flexible attributes to the service. For example, in a travel peer to peer marketplace, a host who owns an accommodation ($t\in \mathcal{D}$) and has a limited (monetary) budget $B$ aims to determine which amenities should be offered in the property such that the costs to offer the amenities
to the host are within the budget $B$, and the gain ($gain(.)$\footnote{\small{In addition to the input set of attributes, the function $gain(.)$ may depend to other variables such as the number of attribute ($n$); one such function is discussed in \S~\ref{sec:gain}.}}) resulting from offering the amenities is maximized.
Formally, our GMFA problem is defined as an optimization problem as shown in Figure~\ref{fig:problmeDef}.

\begin{figure}[!t]
	\small
	\medskip\noindent
	\framebox[\columnwidth]{\parbox{0.9\columnwidth}{ \textbf{\textsc{Gain Maximization over Flexible Attributes (GMFA):}}
	\\ \textit{Given
	a dataset $\mathcal{D}$ with the set of binary flexible attributes $\mathcal{A}$
	where each attribute $A_i\in\mathcal{A}$ is associated with cost $cost[A_i]$,
	a gain function $gain(.)$,
	a budget $B$,
	and a tuple $t\in \mathcal{D}$,
	identify
	a set of attributes $\mathcal{A}^\prime\subseteq \mathcal{A} \backslash \mathcal{A}_t$ ,
	such that
	$$\underset{\forall A_i\in\mathcal{A}^\prime}{\sum} cost[A_i]\leq B$$
	while maximizing
	$$gain(\mathcal{A}_t\cup\mathcal{A}^\prime, \mathcal{D})$$
	}}}
	\caption{GMFA problem definition}
	\label{fig:problmeDef}
\end{figure}

We next discuss the complexity of the general GMFA problem and we show the difficulty of designing an approximation algorithm with a constant approximate ratio for this problem.

\subsection{Computational Complexity}

We prove that GMFA is NP-hard\footnote{\small Please note that GMFA is NP-complete even for the polynomial time gain functions.} by reduction from quadratic knapsack~\cite{quadraticknapsack, garyjohnson} which is NP-complete~\cite{garyjohnson}.
The reduction from the QPK shows that it can be modeled as an instance of GMFA; thus, a solution for QPK cannot not be used to solve GMFA.

\vspace{0.02in}
\noindent{\bf Quadratic 0/1 knapsack (QKP):} Given a set of items $I$, each with a weight $W_i$, a knapsack with capacity $C$, and the profit function $P$, determine a subset of items such that their cumulative weights do not exceed $C$ and the profit is maximized. Note that the profit function $P$ is defined for the choice of individual items but also considers an extra profit that can be earned if two items are selected jointly. In other words, for $n$ items, $P$ is a matrix of $n \times n$, where the diagonal of matrix $P$ ($P_{i,i}$) depicts the profit of item $i$ and an element in $P_{i,j}$ ($i \neq j$) represents the profit of picking item $i$, and item $j$ together.

\begin{theorem}\label{th:npcomplete}
The problem of Gain maximization over flexible attributes (GMFA) is NP-hard.
\end{theorem}

\begin{proof}
The decision version of GMFA is defined as follows: given a decision value $k$, dataset $\mathcal{D}$ with the set of flexible attributes $\mathcal{A}$, associated costs $cost$, a gain function $gain(.)$, a budget $B$, and a tuple $t\in \mathcal{D}$, decide if there is a $\mathcal{A}^\prime\subseteq \mathcal{A} \backslash \mathcal{A}_t$ , such that
$\underset{\forall A_i\in\mathcal{A}^\prime}{\sum} cost[A_i]\leq B$
and $gain(\mathcal{A}_t\cup\mathcal{A}^\prime)\geq k$.\\

We reduce the decision version of quadratic knapsack problem (QKP) to the decision version of GMFA and argue that the solution to QKP exists, if and only if, a solution to our problem exists.
The decision version of QKP, in addition to $I$, $W$, $C$, and $P$, accepts the decision variable $k$ and decides if there is a subset of $I$ that can be accommodated to the knapsack with profit $k$.

A mapping between QKP to GMFA is constructed as follows:
The set of items $I$ is mapped to flexible attributes $\mathcal{A}$, the weight of items $W$ is mapped to the cost of the attributes $cost$, the capacity of the knapsack $C$ to budget $B$, and the decision value $k$ in QKP to the value $k$ in GMFA. Moreover, we set $\mathcal{D}=\{t\}$ and $\mathcal{A}_t=\emptyset$; the gain function, in this instance of GMFA, can be constructed based on the profit matrix $P$ as follows:
$$
gain(\mathcal{A}') = \sum_{\forall A_i\in \mathcal{A}'} P_i + \sum_{\forall A_i\in \mathcal{A}'}\sum_{\forall A_j\neq A_i \in \mathcal{A}'} P_{i,j}
$$
The answer to the QKP is yes (resp. no) if the answer to its corresponding GMFA is yes (resp. no).
\end{proof}

\noindent
Note that GMFA belongs to the NP-complete class only for the $gain(.)$ functions that are polynomial. In those cases the verification of whether or not a given subset $\mathcal{A}^\prime\subseteq \mathcal{A} \backslash \mathcal{A}_t$, has cost less than or equal to $B$ and a gain at least equal to $k$ can be performed in polynomial time.

In addition to the complexity, the reduction from the quadratic knapsack presents the difficulty in designing an approximate algorithm for GMFA.
Rader et. al.~\cite{rader2002quadratic} prove that QKP does not have a polynomial time approximation algorithm with fixed approximation ratio unless P=NP. Even for the cases that $P_{i,j}\geq 0$, it is an open problem whether or not there is an approximate algorithm with a fixed approximate ratio for QKP~\cite{pisinger2007quadratic}. In Theorem~\ref{th:noapprox}, we show that a polynomial approximate algorithm with a fixed approximate ratio for GMFA guarantees a fixed approximate ratio for QKP, and its existence contradicts the result of~\cite{rader2002quadratic}. 
Furthermore, studies on the constrained set functions optimization, such as~\cite{feldman2014constrained}, also admits that maximizing a monotone set function up to an acceptable approximation, even subject to simple constraints is not possible.

\begin{theorem}\label{th:noapprox}
There is no polynomial-time approximate algorithm with a fixed approximate ratio for GMFA unless there is an approximate algorithm with a constant approximate ratio for QKP.
\end{theorem}
\begin{proof}
Suppose there is an approximate algorithm with a constant approximate ratio $\alpha$ for GMFA. Let $app$ be the attribute combination returned by the approximate algorithm and $opt$ be the optimal solution. Since the approximate ratio is $\alpha$:
$$gain(opt)\leq \alpha gain(app)$$
Based on the mapping provided in the proof of Theorem~\ref{th:npcomplete}, we first show that $opt$ is the corresponding set of items for $opt$ in the optimal solution QKP. 
If there is a set of items for which the profit is higher than $opt$, due to the choice of $gain(.)$ in the mapping, its corresponding attribute combination in GMFA has a higher $gain$ than $opt$, which contradicts the fact that $opt$ is the optimal solution of GMFA.
Now since $gain(opt)\leq \alpha gain(app)$:
\begin{align}
\nonumber
&\sum_{\forall A_i\in opt} P_i + \sum_{\forall A_i\in opt}\sum_{\forall A_j\neq A_i\in opt} P_{i,j} \leq\\
&\nonumber \alpha \big( \sum_{\forall A_i\in app} P_i +  \sum_{\forall A_i\in app}\sum_{\forall A_j\neq A_i\in app} P_{i,j} \big)
\end{align}
Thus, the profit of the optimal set of items ($opt$) is at most $\alpha$ times the profit of the set of items ($app$) returned by the approximate algorithm, giving the approximate ratio of $\alpha$ for the quadratic knapsack problem.
\end{proof}
\section{Exact Solution}\label{sec:exact}
Considering the negative result of Theorem~\ref{th:noapprox}, we turn our attention to the design of an exact algorithm for the GMFA problem; even though this algorithm will be exponential in the worst case, we will demonstrate that is efficient in practice.  
In this section, our focus is on providing a solution for GMFA over any {\em monotonic} gain function.
A gain function $gain(.)$ is monotonic, if given two set of attributes $\mathcal{A}_i$ and $\mathcal{A}_j$ where $\mathcal{A}_j\subset\mathcal{A}_i$, $gain(\mathcal{A}_j,\mathcal{D}))\leq gain(\mathcal{A}_i,\mathcal{D}))$.
As a result, {\em this section provides a general solution that works for any monotonic gain function}, no matter 
how and based on what data it is designed.
In fact considering a non-monotonic function for gain is not reasonable here, because adding more attributes to a tuple (service) should not decrease the gain.
For ease of explanation, we first provide the following definitions and notations. Then we discuss an initial solution in \S~\ref{subsec:general-warmup}, which leads to our final algorithm in \S~\ref{sec:generalsolution}.

\vspace{-0.05in}
\begin{definition}{{\bf Lattice of Attribute Combination:}} \label{def:lattice}
Given an attribute combination $\mathcal{A}_i$, the lattice of $\mathcal{A}_i$ is defined as $\mathcal{L}_{\mathcal{A}_i} = (V,E)$,
where the nodeset $V$, depicted as $V(\mathcal{L}_{\mathcal{A}_i})$, corresponds to the set of all subsets of $\mathcal{A}_i$;
thus $\forall \mathcal{A}_j\subseteq \mathcal{A}_i$, there exists a one to one mapping between each $v_j\in V$ and each $\mathcal{A}_j$.
Each node $v_j$ is associated with a $\small{\textsf{bit representative}}$ $\mathcal{B}(v_j)$ of length $m$ in which bit $k$ is $1$ if $A_k\in \mathcal{A}_j$ and $0$ otherwise.
For consistency, for each node $v_j$ in $V(\mathcal{L}_{\mathcal{A}_i})$, the index $j$ is the decimal value of $\mathcal{B}(v_j)$.
Given the bit representative $\mathcal{B}(v_j)$ we define function $v(\mathcal{B}(v_j))$ to return $v_j$.
In the lattice an edge $\langle v_j,v_k \rangle \in E$ exists if $\mathcal{A}_k \subset \mathcal{A}_j$ and $\mathcal{B}(v_k)$, $\mathcal{B}(v_j)$ differ in only one bit. Thus, $v_j$ (resp. $v_k$) is parent (resp. child) of $v_k$ (resp. $v_j$) in the lattice.
For each node $v_j\in V$, $\small{\textsf{level}}$ of $v_j$, denoted by $\ell(v_j)$, is defined as the number of 1's in the bit representative of $v_j$.
In addition, every node $v_j$ is associated with a cost defined as $cost(v_j) = \sum\limits_{\forall A_k \in \mathcal{A}_j}cost[A_k]$.
\end{definition}

\vspace{-0.08in}
\begin{definition}{{\bf Maximal Affordable Node:}}\label{def:maximalaffordable}
A node $v_i\in V(\mathcal{L}_\mathcal{A})$ is affordable iff $\sum\limits_{\forall A_k\in \mathcal{A}} cost$ $[A_k ]\,\leq B$; otherwise it is unaffordable.
An affordable node $v_i$ is maximal affordable iff $\forall$ nodes $v_j$ in parents of $v_i$, $v_j$ is unaffordable. 
\end{definition}

\noindent
{\bf Example 1:} As a running example throughout the paper, 
consider $\mathcal{D}$ as shown in Figure~\ref{fig:sampleset}, defined over the set of attributes $A=\{A_1$:\texttt{Breakfast}, $A_2$:\texttt{TV}, $A_3$:\texttt{Internet}, $A_4$:\texttt{Washer}$\}$ with cost to provide these attributes as $cost=[1000, 300, 250, 700]$. Assume the budget is $B=1300$ and that the property $t$ does not offer these attributes/amenities, i.e., $\mathcal{A}_t=\emptyset$.\\ Figure~\ref{fig:l4} presents $\mathcal{L}_\mathcal{A}$ over these four attributes. 
The bit representative for the highlighted node $v_{10}$ in the figure is $\mathcal{B}(v_{10})=1010$ representing the set of attributes
$\mathcal{A}_{10}=\{A_1$:\texttt{Breakfast}, $A_3$: \texttt{Internet}$\}$;
The level of $v_{10}$ is $\ell(v_{10})=2$, and it is the parent of nodes $v_2$ and $v_8$ with the bit representatives $0010$ and $1000$.
Since $B=1300$ and the cost of $v_2$ is $cost(v_{2}) = 250$, $v_{2}$ is an affordable node; however, since its parent $v_{10}$ the cost $cost(v_{10}) = 1250$ and is affordable, $v_{2}$ is not a maximal affordable node. $v_{11}$ and $v_{14}$ with bit representatives $\mathcal{B}(v_{11})=1011$ and $\mathcal{B}(v_{14})=1110$, the parents of $v_{10}$, are unaffordable; thus $v_{10}$ is a maximal affordable node.

\begin{figure}[t]
\centering
\includegraphics[width=0.49\textwidth]{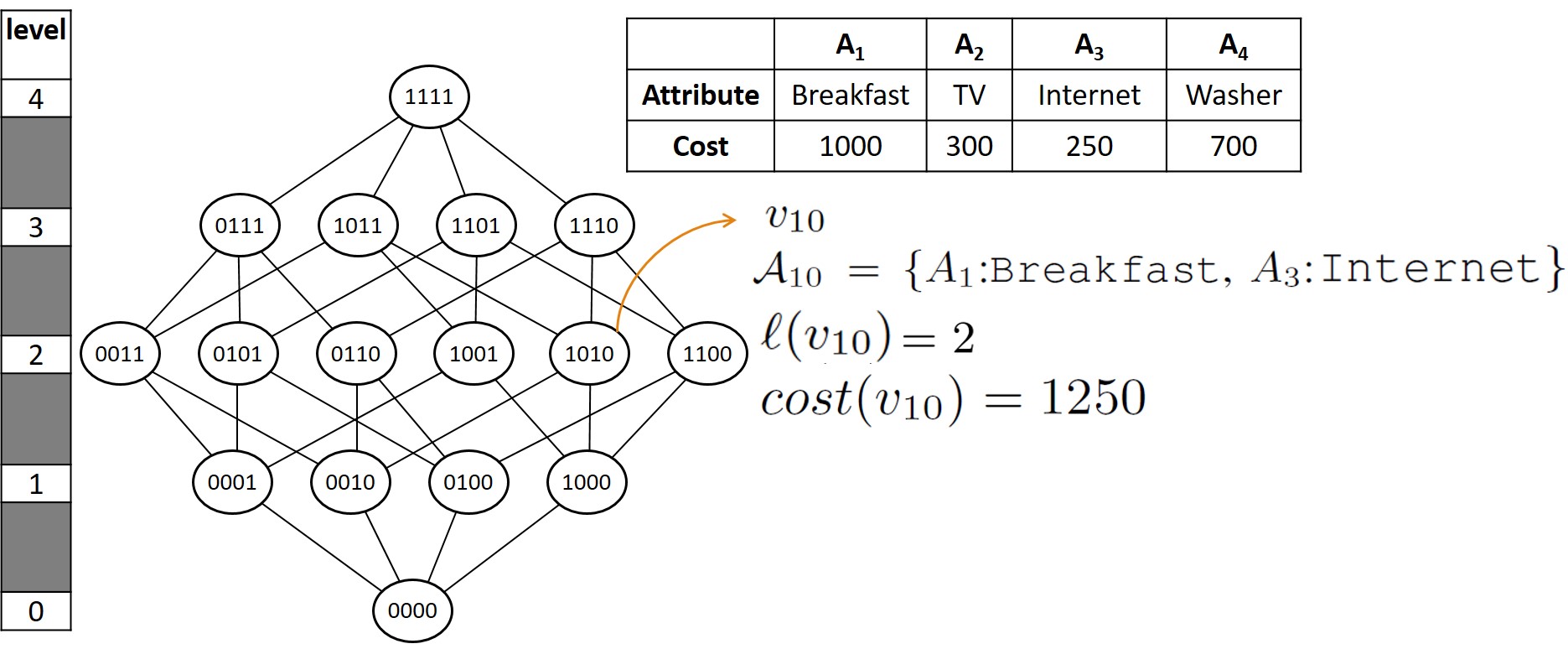}
\caption{Illustration of $\mathcal{L}_\mathcal{A}$ for Example 1} \label{fig:l4}
\end{figure}

A baseline approach for the GMFA problem is to examine all the $2^m$ nodes of $\mathcal{L}_\mathcal{A}$.
Since for every node the algorithm needs to compute the gain, it's running time is in $\Omega(m2^m\mathcal{G})$,
where $\mathcal{G}$ is the computation cost associated with the $gain(.)$ function.

As a first algorithm, we improve upon this baseline by leveraging the monotonicity of the gain function, which enables us to prune some of the branches in the lattice while searching for the optimal solution. This algorithm is described in the next subsection as improved GMFA (I-GMFA). Then, we discuss drawbacks and propose a general algorithm for the GMFA problem in Section~\ref{sec:generalsolution}.

\subsection{I-GMFA}\label{subsec:general-warmup}
An algorithm for GMFA can identify the maximal affordable nodes and return the one with the maximum gain.
Given a node $v_i$, due to the monotonicity of the gain function, for any child $v_j$ of $v_i$, $gain(\mathcal{A}_i)\geq gain(\mathcal{A}_j)$.
Consequently, when searching for an affordable node that maximizes gain, one can ignore all the nodes that are not maximal affordable.
Thus, our goal is to efficiently identify the maximal affordable nodes, while pruning all the nodes in their sublattices.
Algorithm~\ref{alg:warmup} presents a top-down\footnote{\small{One could design a bottom-up algorithm that starts from $v_0$ and keeps ascending the lattice, in BFS manner and stop at the maximal affordable nodes. We did not include it due to its similarity to the top-down approach.}} BFS (breadth first search) traversal of $\mathcal{L}_{\mathcal{A}\backslash \mathcal{A}_t}$ starting from the root of the lattice (i.e., $v(\mathcal{A}\backslash \mathcal{A}_t)$).
To determine whether a node should be pruned or not, the algorithm checks if the node has an affordable parent, and if so, prunes it.
In Example 1, since $cost(v_7)$ ($\mathcal{A}_7=\{A_2$:\texttt{TV}, $A_3$:\texttt{Internet}, $A_4$:\texttt{Washer}$\}$) is $1250 < B$ (and it does not have any affordable parents), $v_7$ is a maximal affordable node; thus the algorithm prunes the sublattice under it.
For the nodes with cost more than $B$, the algorithm generates their children and if not already in the queue, adds them (lines 14 to 16).

\begin{algorithm}[h!]
\caption{{\bf  I-GMFA} \\
         {\bf Input:} Database $\mathcal{D}$ with attributes $\mathcal{A}$, Budget $B$, and Tuple $t$
        }
\begin{algorithmic}[1]
\label{alg:warmup}
\STATE feasible $= \{\}$
\STATE maxg$=0$, best$=$Null
\STATE Enqueue$(v(\mathcal{A}\backslash \mathcal{A}_t))$
\WHILE {$queue$ is not empty}
    \STATE  {$v =$ Dequeue$()$}
    \STATE $g = gain(\mathcal{A}_t\cup \mathcal{A}_v,\mathcal{D})$
    \IF{$g\leq$ maxg {\bf or} $\exists p\in$ parents of $v$ s.t. $p\in $feasible}
        \STATE {\bf continue}
    \ENDIF
    \IF{$ cost(v)\,\leq B$}
        \STATE feasible =feasible $\cup \{ v \}$
        \STATE maxg$=g$, best$=\mathcal{A}_v$
    \ELSE
        \FOR{$v_j$ in children of $v$}
            \STATE {\bf if} $v_j$ is not in $queue$ {\bf then} Enqueue$(v_j)$ 
        \ENDFOR
    \ENDIF
\ENDWHILE
\STATE {\bf return} (best,maxg)
\end{algorithmic}
\end{algorithm}

Algorithm~\ref{alg:warmup} is in $\Omega(\mathcal{G})$, as in some cases after checking a constant number of nodes (especially when the root itself is affordable) it prunes the rest of the lattice. 
Next, we discuss the drawbacks of Algorithm~\ref{alg:warmup} which lead us to a general solution.

\subsection{General GMFA Solution}\label{sec:generalsolution}
Algorithm~\ref{alg:warmup} conducts a BFS over the lattice of attribute combinations which make the time complexity of the algorithm dependent on the number of edges of the lattice. For every node in the lattice, Algorithm~\ref{alg:warmup} may generate all of its children and parents. Yet for every generated child (resp. parent), the algorithm checks if the queue (resp. feasible set) contains it. Moreover, it stores the set of feasible nodes to determine the feasibility of their parents. In this section, we discuss these drawbacks in detail and propose a general approach as algorithm {\bf G-GFMA}.

\subsubsection{The problem with multiple children generation}
\label{sec:multipleChild}
Algorithm~\ref{alg:warmup} generates all children of an unaffordable node. Thus, if a node has multiple unaffordable parents,
Algorithm~\ref{alg:warmup} will generate the children multiple times, even though they will be added to the queue once.
In Figure~\ref{fig:l4}, node $v_9$ ($\mathcal{A}_9 = \{A_1$:\texttt{Breakfast}, $A_4$:\texttt{Washer}$\}$ and $\mathcal{B}(v_9) = 1001$)) will be generated twice by the unaffordable parents $v_{11}$ and $v_{13}$ with the bit representatives 
$\mathcal{B}(v_{11})=1011$ and $\mathcal{B}(v_{13})=1101$;
the children will be added to the queue once, while subsequent attempts for adding them will be ignored as they are already in the queue.

A natural question is whether it is possible  to design a strategy that for each level $i$ in the lattice, (i) make sure we generate all the non-pruned children and (ii) guarantee to generate children of each node only once.

\noindent
\textbf{Tree construction:}
To address this, we adopt the one-to-all broadcast algorithm in a hypercube~\cite{bertsekas1991optimal} constructing a tree that guarantees to generate each node in $\mathcal{L}_\mathcal{A}$ only once.
As a result, since the generation of each node is unique, further checks before adding the node to the queue are not required.
The algorithm works as following: 
Considering the bit representation of a node $v_i$,
let $\rho(\mathcal{B}(v_i))$ be the right-most $0$ in $\mathcal{B}(v_i)$.
The algorithm first identifies $\rho(\mathcal{B}(v_i))$;
then it complements the bits in the right side of $\rho(\mathcal{B}(v_i))$ one by one to generate the children of $v_i$.
Figure~\ref{fig:l4t} demonstrates the resulting tree for the lattice of Figure~\ref{fig:l4} for this algorithm.
For example, consider the node $v_3$ ($\mathcal{B}(v_3)=0011$) in the figure; $\rho(0011)$ is $2$ (for attribute $A_2$).
Thus, nodes $v_1$ and $v_2$ with the bit representatives $\mathcal{B}(v_1)=0001$ and $\mathcal{B}(v_2)=0010$ are generated as its children.

\begin{figure*}[!ht]
	\hspace{-5mm}
	\begin{minipage}[t]{0.3\linewidth}
		\centering
		\includegraphics[width=0.7\textwidth]{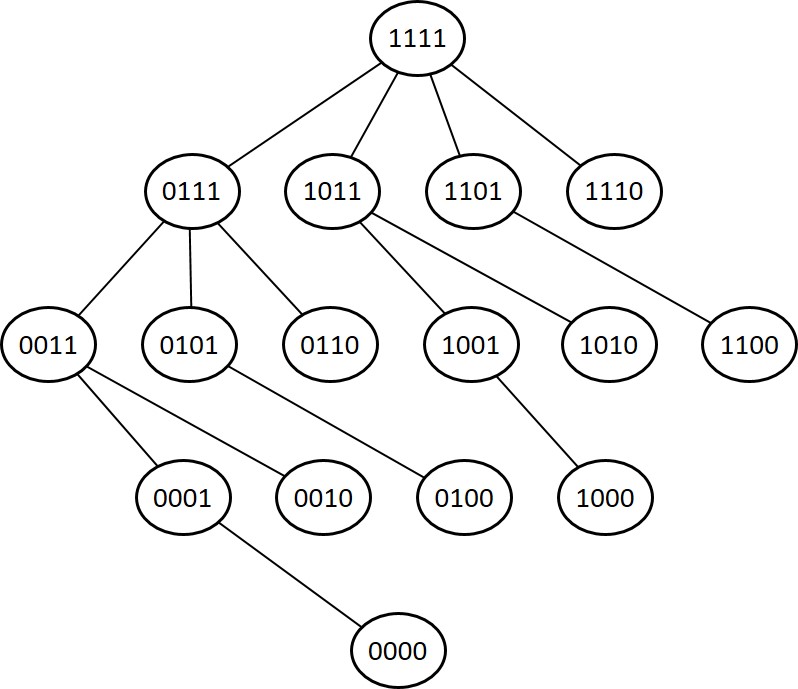}
		\caption{Illustration of tree construction for Figure~\ref{fig:l4}}
		\label{fig:l4t}
	\end{minipage}
	\hspace{1mm}
	\begin{minipage}[t]{0.3\linewidth}
		\centering
		\includegraphics[width=0.95\textwidth]{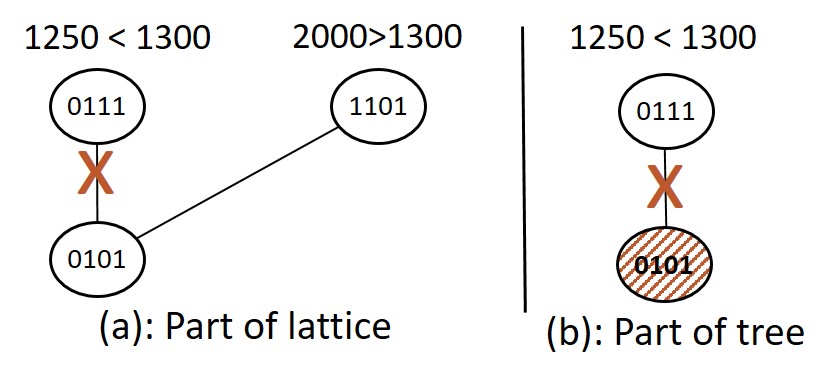}
		\caption{An Example of pruning in Lattice v.s. Tree}
		\label{fig:lfact1}
	\end{minipage}	
	\hspace{1mm}
		\begin{minipage}[t]{0.3\linewidth}
		\centering
		\vspace{-23.5mm}\includegraphics[width=1.2\textwidth]{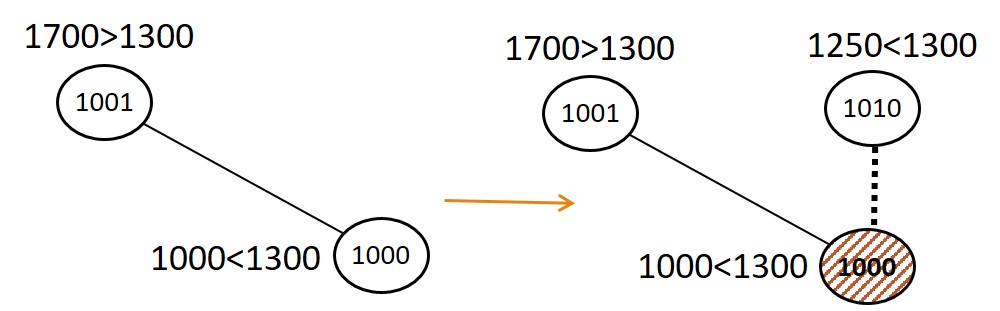}
		\vspace{-1 mm}
		\caption{Checking all parents in lattice to decide if a node is maximal affordable.}
		\label{fig:lfact2}
	\end{minipage}
\end{figure*}


As shown in Figure~\ref{fig:l4t}, i) the children of a node are generated once, and ii) 
all nodes in $\mathcal{L}_\mathcal{A}$ are generated; that is because every node $v_i$ has one (and only one) parent in the tree structure, identified by flipping the bit $\rho(\mathcal{B}(v_i))$ in $\mathcal{B}(v_i)$ to one.
We use parent$_T(v_i)$ to refer to the parent of the node $v_i$ in the tree structure.
For example, for $v_3$ in Figure~\ref{fig:l4}, since $\rho(0101)$ is $3$, its parent in the tree is parent$_T(v_3) = v_7$ ($\mathcal{B}(v_7)=0111$).
Also, note that in order to identify $\rho(\mathcal{B}(v_i))$ there is no need to search in $\mathcal{B}(v_i)$ to identify it.
Based on the way $v_i$ is constructed, $\rho(\mathcal{B}(v_i))$ is the bit that has been flipped by its parent to generate it. 

Thus, by transforming the lattice to a tree, some of the nodes that could be generated by Algorithm~\ref{alg:warmup}, will not be generated in the tree.
Figure~\ref{fig:lfact1} represents an example where the node $v(0101)$ will be generated in the lattice (Figure~\ref{fig:lfact1}(a)) but it will be immediately pruned in the tree (Figure~\ref{fig:lfact1}(b)).
In the lattice, node $v(0101)$ will be generated by the unaffordable parent $v(1101)$, whereas in the tree, Figure~\ref{fig:lfact1}(b), will be pruned, as parent$_T(v(0101)) = v(0111)$ is affordable.
According to Definition~\ref{def:maximalaffordable}, a node $v_i$ that has at least one affordable parent is not maximal affordable.
Since in these cases parent$_T(v_i)$ is affordable, there is no need to generate them.
Note that we only present one step of running Algorithm~\ref{alg:warmup} in the lattice and the tree. Even though node $v(0101)$ is generated (and added to the queue) in Algorithm~\ref{alg:warmup}, it will be pruned in the next step (lines 7 to 9 in Algorithm~\ref{alg:warmup}).

%

\subsubsection{The problem with checking all parents in the lattice}
The problem with multiple generations of children has been resolved by constructing a tree. The pruning strategy is to stop traversing a node in the lattice when a node has at least one affordable parent. I
n Figure~\ref{fig:lfact1}, if for a node $v_i$, parent$_T(v_i)$ is affordable, $v_i$ will not be generated.
However, this does not imply that if a $v_i$ is generated in the tree it does not have an affordable parent in the lattice.
For example, consider $v_{8}$ ($\mathcal{B}(v_{8}) = 1000$ and $\mathcal{A}_{8} = \{A_1$:\texttt{Breakfast}$\}$) in Figure~\ref{fig:l4t}.
We enlarge that part of the tree in Figure~\ref{fig:lfact2}.
As presented in Figure~\ref{fig:lfact2}, parent$_T(v_{8}) = 1001$ is unaffordable and thus $v_8$ is generated in the tree.
However, by consulting the lattice, $v_{8}$ has the affordable parent $v_{10}=v(1010)$ ($\mathcal{A}_{14}=\{A_1$:\texttt{Breakfast}, $A_2$:\texttt{Internet}$\}$); thus $v_{8}$ is not maximal affordable.


In order to decide if an affordable node is maximal affordable, one has to check all its parents in the lattice (not the tree). If at least one of its parents in the lattice is affordable, it is not maximal affordable. Thus, even though we construct a tree to avoid generating the children multiple times, we may end up checking all edges in the lattice since we have to check the affordability of all parents of a node in the lattice. To tackle this problem we exploit the monotonicity of the cost function to construct the tree such that for a given node we only check the affordability of the node's parent in the tree (not the lattice).  

In the lattice, each child has one less attribute than its parents. Thus, for a node $v_i$, one can simply determine the parent with the minimum cost (cheapest parent) by considering the cheapest attribute in $\mathcal{A}$ that does not belong to $\mathcal{A}_i$.
In Figure~\ref{fig:lfact2}, the cheapest parent of $v_{8} = v(1000)$ is $v_{10}=v(1010)$ because $A_3:$\texttt{Internet} is the cheapest missing attribute in $\mathcal{A}_{8}$.
The key observation is that, for a node $v_i$, if the parent with minimum cost is not affordable, none of the other parents is affordable; on the other hand, if the cheapest parent is affordable, there is no need to check the other parents as this node is not maximal affordable.
In the same example, $v_{8}$ is not maximal affordable since its cheapest parent $v_{10}$ has a cost less than the budget, i.e. ($1250 < 1300$).

Consequently, one only has to identify the least cost missing attribute and check if its cost plus the cost of attributes in the combination
is at most $B$. Identifying the missing attribute with the smallest cost is in $O(m)$.
For each node $v_i$, $\rho(\mathcal{B}(v_i))$ is the bit that has been flipped by parent$_T(v_i)$ to generate it.
For example, consider $v_1 = v(0001)$ in Figure~\ref{fig:l4t}; since $\rho(0001)=3$, parent$_T(v_1)=v(0011)$.
We can use this information to reorder the attributes and instantly get the cheapest missing attribute in $\mathcal{A}_i$.
The key idea is that if we originally {\em order the attributes from the most expensive to the cheapest}, 
$\rho(\mathcal{B}(v_i))$ is the index of the cheapest attribute.
Moreover, adding the cheapest missing attribute generates parent$_T(v_i)$.
Therefore, if the attributes are sorted on their cost in descending order, a node with an affordable parent in the lattice will never be generated
in the tree data structure.
Consequently, after presorting the attributes, there is no need to check if a node in the queue has an affordable parent.

Sorting the attributes is in $O(m\log (m))$.
In addition, computing the cost of a node is thus performed in constant time, using the cost of its parent in the tree.
For each node $v_i$, $cost(v_i)$ is $cost(\mbox{parent}_T(v_i))-cost[A_{\rho(\mathcal{B}(v_i))}]$.
Applying all these ideas the final algorithm is in $\Omega\big(\max(m\log (m),\mathcal{G})\big)$ and $O(2^m\mathcal{G})$.
Algorithm~\ref{alg:general} presents the pseudo-code of our final approach, {\bf G-GMFA}.

\begin{algorithm}[!h]
\caption{{\bf G-GMFA} \\
         {\bf Input:} Database $\mathcal{D}$ with attributes $\mathcal{A}$, Budget $B$, and Tuple $t$
        }
\begin{algorithmic}[1]
\label{alg:general}
\STATE $\mathcal{A}' = \mathcal{A}\backslash \mathcal{A}_t\,$; $m' = len(\mathcal{A}')$
\STATE sort $\mathcal{A}'$ descendingly on cost;
\STATE maxg $= 0$
\STATE Enqueue$\big( v(\mathcal{A}'),-1,\overset{m'-1}{\underset{i=0}{\sum}}cost[\mathcal{A}'[i]])$
\WHILE {$queue$ is not empty}
    \STATE  {$(v,i,\sigma )=$ Dequeue$()$}
    \IF{$\sigma\leq B$}
        \STATE $g = gain(\mathcal{A}_t\cup \mathcal{A}_v,\mathcal{D})$ \label{line:G-Gmfa-gain}
        \STATE {\bf if} $g>$maxg {\bf then } maxg $= g\,$; best$=\mathcal{A}_v$
    \ELSE 
        \FOR{$j$ from $(i+1)$ to $(m'-1)$}
            \STATE Enqueue$\big(v(\mathcal{A}_v\backslash \mathcal{A}'[j]),j,\sigma - cost[\mathcal{A}'[j]] \big)$
        \ENDFOR
    \ENDIF
\ENDWHILE
\STATE {\bf return} (best,maxg)
\end{algorithmic}
\end{algorithm}

\section{Gain Function Design} \label{sec:gain}
As the main contribution of this paper, in \S~\ref{sec:exact}, we
proposed a general solution that works for any arbitrary monotonic gain function.
We conducted our presentation for a generic gain function because the design of the gain function is application specific and depends on the available information. 
The application specific nature of the gain function design, motivated the consideration of the generic gain function, instead of promoting a specific function.
Consequently, applying any monotonic gain function in Algorithm~\ref{alg:general} is as simple as calling it in line~\ref{line:G-Gmfa-gain}.

In our work, the focus is on understanding which subsets of attributes are attractive to users.
Based on the application, in addition to the dataset $\mathcal{D}$, some extra information (such as query logs and user ratings) may be available that help in understanding the desire for combinations of attributes and could be a basis for the design of such a function.
However, such comprehensive information that reflect user preferences are not always available.
Consider a third party service for assisting the service providers. Such services have a limited view of the data~\cite{queryreranking} and may only have access to the dataset tuples. An example of such third party services is AirDNA~\footnote{\small{www.airdna.co}} which is built on top of AirBnB.
Therefore, instead on focusing on a specific application and assuming the existence of extra information, in the rest of this section, we focus on a simple, yet compelling variant of a practical gain function that only utilizes
the existing tuples in the dataset to define the notion of gain in the absence of other sources of information.
We provide a general discussion of gain functions with extra information in Appendix~\ref{subsec:dis-userpref}.

\subsection{Frequent-item Based Count (FBC)}

In this section, we propose a practical gain function that only utilizes the existing tuples in the dataset. It hinges on the observations that the bulk of market participants are expected to behave rationally. Thus, goods on offer are expected to follow
a basic supply and demand principle.
For example, based on the case study provided in \S~\ref{subsec:casestudy}, while many of the properties in Paris offer \texttt{Breakfast}, offering it is not popular in New York City. This indicates a relatively high demand for such an amenity in Paris and relatively low demand in New York City.
As another example, there are many accommodations that provide \texttt{washer}, \texttt{dryer}, and \texttt{iron} together; 
providing \texttt{dryer} without a \texttt{washer} and \texttt{iron} is rare. This reveals a need for the combination of these
 attributes.
Utilizing this intuition, we define a frequent node in $\mathcal{L}_\mathcal{A}$ as follows:

\begin{definition}{{\bf Frequent Node:}}\label{def:fqc}
Given a dataset $\mathcal{D}$, and a threshold $\tau\in (0,1]$, a node $v_i\in V(\mathcal{L}_\mathcal{A})$ is frequent if and only if the number of tuples in $\mathcal{D}$ containing the attributes $\mathcal{A}_i$ is at least 
$\tau$ times $n$, i.e., $|Q(v_i,\mathcal{D})|\geq \tau n$.\footnote{For simplicity, we use $Q(v,\mathcal{D})$ to refer to $Q(\mathcal{A}(v),\mathcal{D})$.}
\end{definition}

For instance, in Example 1 let $\tau$ be $0.3$. In Figure~\ref{fig:sampleset}, $v_3=v(0011)$ is frequent because Accom. 2, Accom. 5, and Accom. 9 contain the attributes $\mathcal{A}_3 = \{A_3$:\texttt{Internet}, $A_4$:\texttt{Washer} $\}$; thus $|Q(v_3,\mathcal{D})|=3\geq 0.3\times 10$. However, since $|Q(v_{11} = v(1011)),\mathcal{D})|$ is $1 < 0.3\times 10$, 
$v_{11}$ is not frequent. The set of additional notation, utilized in \S~\ref{sec:gain} is provided in Table~\ref{table:notations2}.

Consider a tuple $t$ and a set of attributes
$A'\subseteq \mathcal{A}\backslash \mathcal{A}_t$ to be added to $t$.
Let $\mathcal{A}_i$ be $\mathcal{A}_t\cup\mathcal{A}'$ and $v_i$ be $v(\mathcal{A}_i)$.
After adding $\mathcal{A}'$ to $t$, for any node $v_j$ in $\mathcal{L}_{\mathcal{A}_i}$, $t$ belongs to $Q(v_j,\mathcal{D})$.
However, according to Definition~\ref{def:fqc}, only the frequent nodes in $\mathcal{L}_{\mathcal{A}_i}$ are desirable.
Using this intuition, Definition~\ref{def:FBC} provides a practical gain function utilizing nothing but the tuples in the dataset.
 
\begin{definition}{{\bf Frequent-item Based Count (FBC):}}\label{def:FBC}
Given a dataset $\mathcal{D}$, and a node $v_i\in V(\mathcal{L}_\mathcal{A})$, the Frequent-item Based Count (FBC) of $v_i$ is the number of
frequent nodes in $\mathcal{L}_{\mathcal{A}_i}$. Formally
\begin{align}
\mbox{FBC}_\tau (\mathcal{B}(v_i),\mathcal{D}) = |\{v_j\in \mathcal{L}_{\mathcal{A}_i} \, |\, |Q(v_j,\mathcal{D})|\geq \tau \}|
\end{align}
\end{definition}
For simplicity, throughout the paper we use FBC($\mathcal{B}(v_i)$) to refer to FBC$_\tau (\mathcal{B}(v_i),\mathcal{D})$.
In Example 1, consider $v_{15}=v(1111)$. In Figure~\ref{fig:l4intersect}, we have colored the frequent nodes in $\mathcal{L}_{\mathcal{A}_{15}}$. Counting the number of colored nodes in Figure~\ref{fig:l4intersect}, FBC($1111$) is $13$.

Such a definition of a gain function has several advantages, mainly (i) it requires
knowledge only of the existing tuples in the dataset (ii) it naturally
captures changes in the joint demand for certain attribute combinations (iii) it
is robust and adaptive to the underlying data changes.
However, it is known that~\cite{gunopulos2003discovering}, counting the number of frequent itemsets is \#P-complete.
Consequently, counting the number of frequent subsets of a subset of attributes (i.e., counting the number of frequent nodes in $\mathcal{L}_{\mathcal{A}_i}$) is exponential to the size of the subset (i.e., the size of $\mathcal{L}_{\mathcal{A}_i}$).
Therefore, for this gain function, even the verification version of GMFA is likely not solvable in polynomial time.

Thus, in the rest of this section, we design a practical output sensitive algorithm for computing FBC($\mathcal{B}(v_i)$).
Next, we discuss an observation that leads to an effective computation of FBC
and a negative result against such a design.

\begin{table}[!t]
\center
\footnotesize
\vspace{5mm}
\caption{Table of additional notations for \S~\ref{sec:gain}}
\begin{tabular}{|@{}c@{}|@{}c@{}|}
\hline
{\bf Notation}& {\bf Meaning}\\ \hline
$\tau$& The frequency threshold \\ \hline 
$Q(v_i,\mathcal{D})$& The set of tuples in $\mathcal{D}$ that contain the attributes $\mathcal{A}_i$\\ \hline
FBC$(\mathcal{B}(v_i))$& The frequent-item based count of  $\mathcal{B}(v_i)$\\ \hline 
$\mathcal{F}_{v_i}$& The set of maximal frequent nodes in $\mathcal{L}_{\mathcal{A}_i}$\\ \hline
$\mathcal{U}_{v_i}$& The set of frequent nodes in $\mathcal{L}_{\mathcal{A}_i}$\\ \hline
$P$& The pattern $P$ that is a string of size $m$ \\ \hline 
COV$(P)$& The set of nodes covered by the pattern $P$\\ \hline
$k_x(P)$& Number of $X$s in $P$\\ \hline
$G_j$& The bipartite graph of the node $v_j$\\ \hline 
$\mathcal{V}(G_j)$& The set of nodes of $G_j$\\ \hline 
$\mathcal{E}(G_j)$& The set of edges of $G_j$\\ \hline 
$\xi_{k,l}\in \mathcal{E}(G_j)$& The edge from the node $A_k$ to $P_l$ in $\mathcal{E}(G_j)$\\ \hline 
$\xi_j^k$& The adjacent nodes to the $A_k$ in $G_j$\\ \hline 
$\delta_k$& The binary vector showing the nodes $v_j$ where $A_k\in\mathcal{A}_j$\\ \hline  
$\alpha(P_j,G_j)$& Part of $\mathcal{U}_{v_i}$ assigned to $P_j$ while applying Rule 1 on $G_j$\\ \hline 
\end{tabular}
\label{table:notations2}
\end{table}

\subsection{FBC computation -- Initial Observations}\label{subsec:4-1}
Given a node $v_i$, to identify FBC$(\mathcal{B}(v_i))$, the baseline solution traverses the lattice under  $v_i$, i.e.,$\mathcal{L}_{\mathcal{A}_i}$
counting the number of nodes in which more than $\tau$ tuples in the dataset match
the attributes corresponding to $v_i$. Thus, this baseline is always in 
$\theta(n2^{\ell(v_i)})$. 
An improved method to compute FBC of $v_i$, is to start from the bottom of the lattice
$\mathcal{L}_{\mathcal{A}_i}$ and follow the well-known Apriori\cite{apriori} algorithm
discovering the number of frequent nodes. This algorithm utilizes the fact that
any superset of an infrequent node is also infrequent.
The algorithm combines pairs of frequent nodes at level $k$ that share $k-1$ attributes,
to generate the candidate nodes at level $k+1$. It then checks the frequency of candidate pairs at level $k+1$ to identify the frequent nodes of size $k+1$ and
continues until no more candidates are generated. 
Since generating the candidate nodes at level $k+1$ contains combining the frequent nodes at level $k$, this algorithms is in $O(n.$FBC$(\mathcal{B}(v_i))^2)$.

Consider a node $v_i$ which is frequent. In this case, Apriori will generate all the $2^{\ell(v_i)}$ frequent nodes, 
i.e., in par with the baseline solution.
One interesting observation is that if $v_i$ itself is frequent, since all nodes in $\mathcal{L}_{\mathcal{A}_i}$ are also frequent, FBC$(\mathcal{B}(v_i))$ is $2^{\ell(v_i)}$. As a result, in such cases, FBC can be computed in constant time. In Example 1,
since node $v_7$ with bit representative $\mathcal{B}(v_7) = 0111$ is frequent FBC$(0111) = 2^3 = 8$ ($\ell(v_7) = 3$).

This motivates us to compute the number of frequent nodes in a lattice without generating all the nodes.
First we define the set of maximal frequent nodes as follows:
\begin{definition}{{\bf Set of Maximal Frequent Nodes:}} \label{def:maxFreq}
Given a node $v_i$, dataset $\mathcal{D}$, and a threshold $\tau\in (0,1]$,
the set of maximal frequent nodes is the set of frequent nodes in $V(\mathcal{L}_{\mathcal{A}_i})$ that do not have a frequent parent. Formally,
\begin{align}
\mathcal{F}_{v_i}(\tau,\mathcal{D}) =& \{v_j \in V(\mathcal{L}_{\mathcal{A}_i}) |\, |Q(v_j,\mathcal{D})|\geq \tau n \mbox{ and } \\
                                     & \nonumber\forall v_k\in \mbox{parents}(v_j, \mathcal{L}_{\mathcal{A}_i}):\, |Q(v_k,\mathcal{D})|< \tau n\}
\end{align}
\end{definition}
In the rest of the paper, we ease the notation $\mathcal{F}_{v_i}(\tau,\mathcal{D})$ with  $\mathcal{F}_{v_i}$. In Example 1, the set of maximal frequent nodes of $v_{15}$ with bit representative $\mathcal{B}(v_{15}) = 1111$ is $\mathcal{F}_{v_{15}} = \{ v_7, v_{10}, v_{14}\}$, where $\mathcal{B}(v_{7}) = 0111$, $\mathcal{B}(v_{10}) = 1001$, and $\mathcal{B}(v_{14}) = 1110$.

Unfortunately,
unlike the cases where $v_i$ itself is frequent, calculating the FBC of infrequent nodes is challenging. That is because the intersections between the frequent nodes in the sublattices of $\mathcal{F}_{v_i}$ are not empty. 
Due to the space limitations, please find further details on this negative result in Appendix~\ref{sec:fibNeg}.
Therefore, in \S~\ref{sec:4-3}, we propose an algorithm that breaks the frequent nodes in the sublattices of $\mathcal{F}_{v_i}$ into disjoint partitions.

\begin{figure*}[ht]
    \begin{minipage}[t]{0.31\linewidth}
        \centering
        \includegraphics[width=0.7\textwidth]{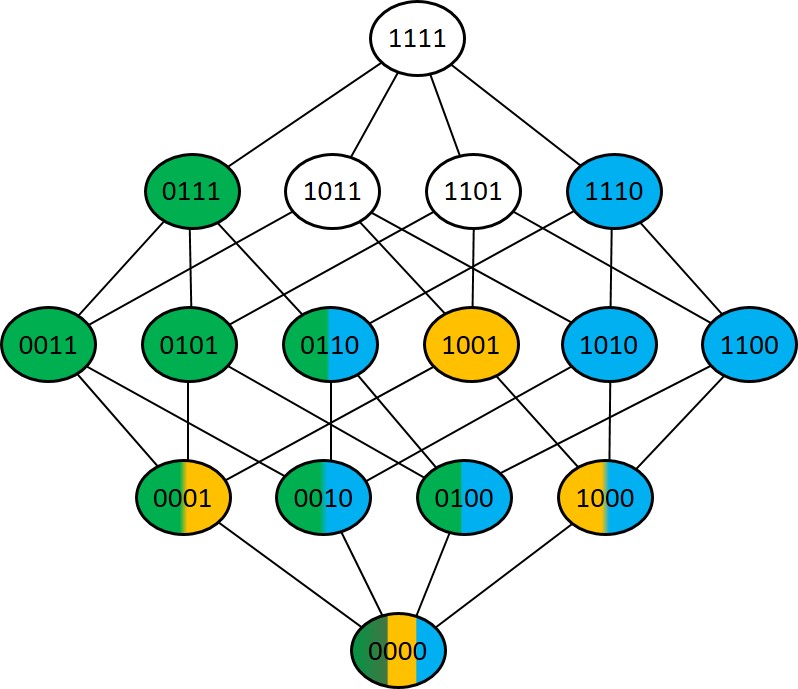}
        \caption{Illustration of the intersection between the children of $C=\{A_1,A_2,A_4\}$}
        \label{fig:l4intersect}
    \end{minipage}
    \hspace{2mm}
    \begin{minipage}[t]{0.31\linewidth}
        \centering
        \includegraphics[width=0.7\textwidth]{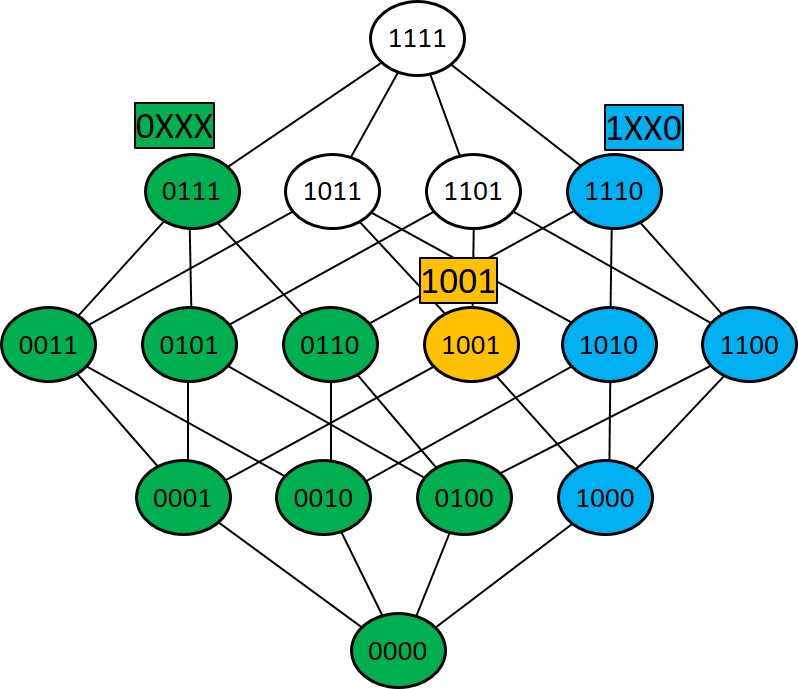}
        \caption{Illustration of sublattice coverage by the tree}
        \label{fig:l4tpartition}
    \end{minipage}
    \hspace{4mm}
    \begin{minipage}[t]{0.3\linewidth}
        \vspace{-23mm}
        \centering
        \begin{small}
		\begin{tabular}{|@{}c@{}|@{}c@{}|@{}c@{}|@{}c@{}|@{}c@{}|}
            \hline
            {\bf id}&{\bf $v_i\in \mathcal{F}_{v_i}$}&{\bf $\mathcal{B}(v_i)$}& {\bf Disjoint Patterns}& {\bf Count}\\ \hline
            $1$&$v_{7}$&$0111$& $ 0XXX$& $2^3$\\ \hline
            $2$&$v_{14}$&$1110$& $ 1XX0$& $2^2$\\ \hline
            $3$&$v_{9}$&$1001$& $ 1\, 0\, 0\, 1$& $2^0$\\ \hline
          \end{tabular}
        \end{small}
        \vspace{8.5mm}
        \caption{FBC($1111$) $=8+4+1=13$}
        \label{fig:l4FBC}
    \end{minipage}
\end{figure*}
\subsection{A negative result on computing FBC using $\mathcal{F}_{v_i}$}\label{sec:fibNeg}
As discussed in \S~\ref{subsec:4-1}, if a node $v_i$ is frequent, all the nodes in $\mathcal{L}_{\mathcal{A}_i}$ are also frequent, FBC($\mathcal{B}(v_i)$) is $2^{\ell(v_i)}$. 
Considering this, given $v_i$, suppose $v_j$ is a node in the set of maximal frequent nodes of $v_i$ (i.e., $v_j\in \mathcal{F}_{v_i}$); the FBC of any node $v_k$ where $\mathcal{A}_k\subseteq \mathcal{A}_j$ can be simply calculated as FBC($\mathcal{B}(v_k)$) $= 2^{\ell(v_k)}$.
In Example 1, node $v_7$ with bit representative $\mathcal{B}(v_7) = 0111$ is in $\mathcal{F}_{v_{15}}$, thus FBC$(0111) = 8$;
for the node $v_3 = v(0011)$ also, 
since $\mathcal{A}_3=\{A_3,A_4\}$ is a subset of $\mathcal{A}_7=\{A_2,A_3,A_4\}$, FBC$(0011)$ is $2^2 = 4$ ($\ell(v_3) = 2$).

Unfortunately, this does not hold for the nodes whose attributes are a superset of a maximal frequent node; calculating the FBC of those nodes is challenging.
Suppose we wish to calculate the FBC of $v_{15}=v(1111)$ in Example 1.
The set of maximal frequent nodes of $v_{15}$ is $\mathcal{F}_{v_{15}} = \{ v_7, v_{10}, v_{14}\}$, where $\mathcal{B}(v_{7}) = 0111$, $\mathcal{B}(v_{10}) = 1001$, and $\mathcal{B}(v_{14}) = 1110$. Figure~\ref{fig:l4intersect} presents the sublattice of each of the maximal frequent nodes in a different color. The nodes in 
$\mathcal{L}_{\mathcal{A}_7}$ are colored green, while the nodes in $\mathcal{L}_{\mathcal{A}_{10}}$ are orange and the nodes in $\mathcal{L}_{\mathcal{A}_{14}}$ are blue.
Several nodes in $\mathcal{L}_{\mathcal{A}_{15}}$ (including $v_{15}$ itself) are not frequent; thus FBC$(1111)$ is less than $2^4$ ($\ell(v_{15})=4$).
In fact, FBC$(1111)$ is equal to the size of the union of colored sublattices. Note that the intersection between the sublattices of the maximal frequent nodes is not empty. Thus, even though for each maximal frequent node $v_j \in \mathcal{F}_{v_i}$, the FBC$(\mathcal{B}(v_{j}))$ is $2^{\ell(v_j)}$, computing the FBC$(\mathcal{B}(v_{i}))$ is not (computationally) simple.
If we simply add the FBC of all maximal frequent nodes in $\mathcal{F}_{v_i}$, we are overestimating FBC$(\mathcal{B}(v_{i}))$, because we are counting the intersecting nodes multiple times. 

More formally, given an infrequent node $v_{i}$ with maximal frequent nodes $\mathcal{F}_{v_i}$, the FBC$(\mathcal{B}(v_{i}))$ is equal to the size of the union of the sublattices of its maximal frequent nodes which utilizing the inclusion$-$exclusion principle is provided by Equation~\ref{eq:setunion}. In this equation, $v_j \in \mathcal{F}_{v_i}$ and $\mathcal{L}_{\mathcal{A}_j}$ is the sublattice of node $v_j$.

\begin{align}\label{eq:setunion}
FBC(\mathcal{B}(v_{i})) &= \big|\underset{v_j \in \mathcal{F}_{v_i}}{\overset{}{\bigcup}} \mathcal{L}_{\mathcal{A}_j}\big| \\
&= \sum\limits_{j=1}^{|\mathcal{F}_{v_i}|} (-1)^{j+1} \Big( \sum\limits_{\forall \mathcal{F}'\subset \mathcal{F}_{v_i},|\mathcal{F}'|=j} \big|\underset{\forall v_k\in \mathcal{F}'}{\bigcap} \mathcal{A}_k\big|\Big) \nonumber
\end{align}

For example, in Figure~\ref{fig:l4intersect}:
\begin{align}
\nonumber \mbox{FBC}(1111) &= \mbox{FBC}(0111) + \mbox{FBC}(1001) + \mbox{FBC}(1110)\\
\nonumber & - \mbox{FBC}(0110)  - \mbox{FBC}(0001) - \mbox{FBC}(1000)\\
\nonumber & + \mbox{FBC}(0000) \\
\nonumber & = 2^3 + 2^2 + 2^3 - 2^2 - 2^1 - 2^1 + 2^0 = 13
\end{align}

\noindent
Computing FBC based on Equation~\ref{eq:setunion} requires to add (or subtract) $\sum\limits_{j=1}^{|\mathcal{F}_{v_i}|} {|\mathcal{F}_{v_i}| \choose j} = 2^{|\mathcal{F}_{v_i}|}-1$ terms, thus its running time is in $\theta(2^{|\mathcal{F}_{v_i}|})$.

\subsection{Computing FBC using $\mathcal{F}_{v_i}$} \label{sec:4-3}
As elaborated in \S~\ref{sec:fibNeg}, it is evident that for a given infrequent node $v_i$, the intersection between the sublattices of its maximal frequent nodes in $\mathcal{F}_{v_i}$ is not empty.
Let $\mathcal{U}_{v_i}$ be the set of all frequent nodes in $\mathcal{L}_{\mathcal{A}_i}$ -- i.e., $\mathcal{U}_{v_i} = \underset{\forall v_j \in \mathcal{F}_{v_i}}{\bigcup} V(\mathcal{L}_{\mathcal{A}_j})$. In Example 1, the $\mathcal{U}_{v_{15}}$ is a set of all colored nodes in Figure~\ref{fig:l4intersect}. 

Our goal is to {\em partition} $\mathcal{U}_{v_i}$ to a collection $\mathcal{S}$ of disjoint sets such that (i) $\underset{\forall S_i\in \mathcal{S}}{\cup}(S_i) = \mathcal{U}_{v_i}$ and (ii) the intersection of the partitions is empty, i.e., $\forall S_i, S_j\in \mathcal{S}$, $S_i\cap S_j = \emptyset$; given such a partition, FBC$(\mathcal{B}(v_i))$ is $\sum\limits_{\forall S_i\in \mathcal{S}}|S_i|$.
Such a partition for Example 1 is shown in Figure~\ref{fig:l4tpartition}, where each color (e.g, blue) represents a set of nodes which is disjoint from the 
other sets designated by different colors (e.g., orange and green). In the rest of this section, we propose an algorithm for such disjoint partitioning.

Let us first define {\em ``pattern''} $P$ as a string of size $m$, where $\forall 1\leq i\leq m$: $P[i]\in\{0,1,X\}$.
Specially, we refer to the pattern generated by replacing all $1$s in $\mathcal{B}(v_i)$ with $X$ as the {\em initial pattern} for $v_i$.
For example, in Figure~\ref{fig:l4tpartition}, there are four attributes, $P=0XXX$ is a pattern (which is the initial pattern for $v_7$). The pattern $P=0XXX$ covers all nodes whose bit representatives start with $0$ (the nodes with green color in Figure~\ref{fig:l4intersect}). More formally, the coverage of a pattern $P$ is defined as follows:

\begin{definition}
Given the set of attributes $\mathcal{A}$ and a pattern $P$, the coverage of pattern $P$ is\footnote{\small{Note that if $P[k]=X$, $A_k$ may or may not 
belong to $\mathcal{A}_i$.}}
\begin{align}
\nonumber
\mbox{COV}(P) = \{v_i\subseteq V(\mathcal{L}_\mathcal{A}) |\forall 1\leq k& \leq m\mbox{ if }P[k]=1\mbox{ then} A_k\in \mathcal{A}_i \\
                                              &\mbox{, and if }P[k]=0\mbox{ then} A_k\notin \mathcal{A}_i\}
\end{align}
\end{definition}
In Figure~\ref{fig:l4intersect}, all nodes with green color are in $COV(0XXX)$ and all nodes with blue color are in $COV(XXX0)$. Specifically, node $v_3$ with bit representative $\mathcal{B}(v_3)=0011$ is in $COV(0XXX)$ because its first bit is $0$. Note that a node may be covered by multiple patterns, e.g., node $v_6$ with bit representative $\mathcal{B}(v_6)=0110$ is in $COV(0XXX)$ and $COV(XXX0)$. We refer to patterns with disjoint coverage as {\em disjoint patterns}. For example, $01XX$ and $11XX$ are {\em disjoint patterns}.

Figure~\ref{fig:l4tpartition}, provides a set of disjoint patterns (also presented in the 4-th column of Figure~\ref{fig:l4FBC}) that partition $\mathcal{U}_{v_{15}}$ in Figure~\ref{fig:l4intersect}. The nodes in the coverage of each pattern is colored with a different color.
Given a pattern $P$ let $x_k(P)$ be the number of $X$s in $P$; the number of nodes covered by the pattern $P$ is $2^{k_x(P)}$. Thus given a set of disjoint patterns $\mathcal{P}$ that partition $\mathcal{U}_{v_i}$, FBC($\mathcal{B}(v_i)$) is simply $\sum\limits_{\forall P_j\in\mathcal{P}} 2^{k_x(P_j)}$.
For example, considering the set of disjoint patterns in Figure~\ref{fig:l4tpartition}, the last column of Figure~\ref{fig:l4FBC} presents the number of nodes in the coverage of each pattern $P_j$ (i.e., $2^{k_x(P_j)}$); thus FBC($1111$) in this example is the summation of the numbers in the last column (i.e., $13$).

In order to identify the set of disjoint patterns that partition $\mathcal{U}_{v_i}$, a baseline solution may need to compare every initial pattern for every node $v_j\in\mathcal{F}_{v_i}$ with all the discovered patterns $P_l$, to determine the set of patterns that are disjoint from $P_l$. As a result, because every pattern covers at least one node in $\mathcal{L}_(\mathcal{A}_i)$, the baseline solution may generate up to FBC($\mathcal{B}(v_i)$)$=|\mathcal{U}_{v_i}|$ patterns and (comparing all patterns together) its running time in worst case is quadratic in its output (i.e., FBC$(\mathcal{B}(v_i))^2$).
As a more comprehensive example, let us introduce Example 2; we will use this example in what follows.

\noindent{\bf Example 2.} 
Consider a case where $\mathcal{A}=\{A_1, \dots , A_{11}\}$ and we want to compute FBC for the root of $\mathcal{L}_\mathcal{A}$, i.e. $v_i = v(\mathcal{A})$. Let $\mathcal{F}_{v_i} $ be $\{ v_{1984},v_{248},v_{499}\}$ as shown in Figure~\ref{fig:E3F}.

\begin{figure}[!h]
        \centering
        \begin{small}
          \begin{tabular}{|@{}c@{}|@{}c@{}|@{}c@{}|@{}c@{}|}
            \hline
            {\bf id}&{\bf $v_i\in \mathcal{F}_{v_i}$}&{\bf $\mathcal{B}(v_i)$}&{\bf Initial Patterns}\\ \hline
            $1$&$v_{1984}$&$11111000000$&$ P_1 = XXXXX000000$\\ \hline
            $2$&$v_{248}$&$ 00011111000$&$ P_2 = 000XXXXX000$\\ \hline
            $3$&$v_{499}$&$ 00111110011$&$ P_3 = 00XXXXX00XX$\\ \hline
          \end{tabular}
        \end{small}
        \caption{$\mathcal{F}_{v_i}$ in Example 2}
        \label{fig:E3F}
\end{figure}

In Example 2, consider the two initial patterns $P_1$ and $P_2$, for the nodes $v_{1984}$ and $v_{248}$, respectively. Note that for each initial pattern $P$ for a node $v_j$, $COV(P)$ is equal to $\mathcal{L}_{\mathcal{A}_j}$.
In order to partition the space in $COV(P_2)\setminus COV(P_1)$, the baseline solution generates the three patterns shown in Figure~\ref{fig:cov2-1}.
\begin{figure}[!h]
        \centering
        \small
        \begin{align}
        \nonumber COV(P_2)\setminus COV(P_1) &= P_{2.1}: 000XX\underline{1XX}000\\
        \nonumber &, \hspace{0.1in} P_{2.2}: 000XX\underline{01X}000\\
        \nonumber &, \hspace{0.1in} P_{2.3}: 000XX\underline{001}000
        \end{align}
        \caption{Patterns of $COV(P_2)\setminus COV(P_1)$ in Example 2}
        \label{fig:cov2-1}
\end{figure}

Yet for the initial pattern $P_3$ (for the node $v_{449}$), in addition to $P_1$, the baseline needs to compare it against $P_{2.1}$, $P_{2.2}$, and $P_{2.3}$ (while at each comparison a pattern may decompose to multiple patterns) to guarantee that the generated patterns are disjoint.
Thus, 

it is desirable to partition $\mathcal{U}_{v_i}$ in a way to generate only one pattern
per maximal frequent node. .

Consider the disjoint patterns generated for $COV(P_2)\setminus COV(P_1)$ in Example 2. As shown in Figure \ref{fig:cov2-1} these patterns differ in the underlined attributes $A_6$, $A_7$, and $A_8$. These attributes are the ones for which $P_2[k]=X$ and $P_1[k]=0$ (column 3 in Figure~\ref{fig:E3F}). Clearly, if all these attributes are zero, then they will be in $COV(P_1)$. Thus, if we mark these attributes and enforce that at least one of these attributes should be $1$, there is no need to generate three patterns.

To do so, for each maximal frequent node $v_{j'} = \mathcal{F}_{v_i}[j]$ with initial pattern $P_j$, we create a {\em bipartite graph} $G_j(\mathcal{V}, \mathcal{E})$ as follows:
\begin{enumerate}
\item For each attribute $A_k$ that belongs to $\mathcal{A}_{j'}$, we add a node $A_k$ in one part (top level) of $G_j$; i.e., $\forall A_k\in \mathcal{A}_{j'}:\, A_k\in \mathcal{V}(G_j)$.
\item For each $1\leq l < j$, we add the node $P_l$ ( the pattern for the node $v_{l'} =\mathcal{F}_{v_i}[l]$) in the other part (bottom level) of $G_j$, i.e., $\forall l\in [1,j):\, P_l\in \mathcal{V}(G_j)$.
\item For each $1\leq l < j$ and each attribute $A_k$ in $\mathcal{A}_{j'}$ that $A_k$ does not belong to $\mathcal{A}_{l'}$, i.e. $P_l[k] = 0$, we add the edge $\xi_{k,l}$ from $A_k$ to $P_l$ in $G_j$; i.e., $\forall l\in [1,j), A_k\in\mathcal{A}_{j'}$ where $P_l[k] = 0:\, \xi_{k,l}\in \mathcal{E}(G_j)$. 
\end{enumerate}

\begin{figure*}[t]
\begin{minipage}[t]{0.46\linewidth}
        \centering
        \includegraphics[width=1.05\textwidth]{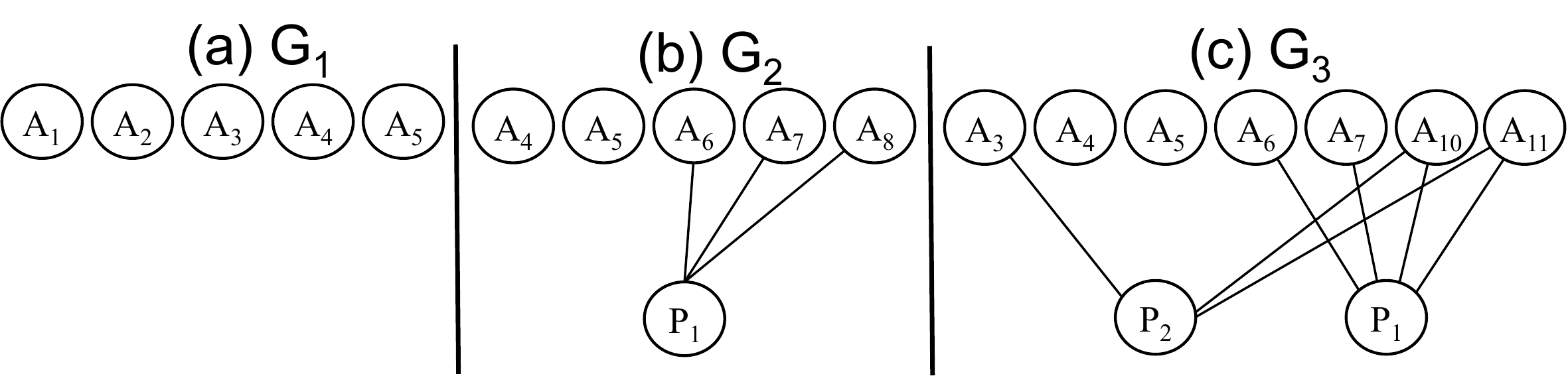}
        \caption{(a) $G_1$ for $v_{1984}=\mathcal{F}_{v_i}[1]$ with pattern $P_1$, (b) $G_2$ for $v_{248}=\mathcal{F}_{v_i}[2]$ with pattern $P_2$, (c) $G_3$ for $v_{499}=\mathcal{F}_{v_i}[3]$ with pattern $P_3$.}
        \label{fig:bipartite}
    \end{minipage}
    \hspace{3mm}
    \begin{minipage}[t]{0.22\linewidth}
        \centering
        \includegraphics[width=0.89\textwidth]{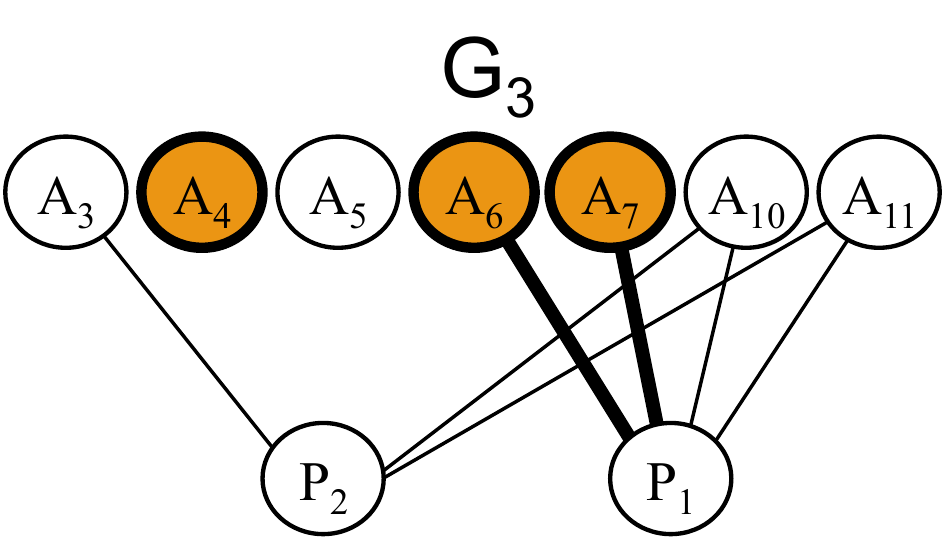}
        \caption{$v_1 = 00010110000$ will not be assigned to $P_3$, $v_1 \notin \alpha(P_3,G_3)$. }
        \label{fig:E3bipartite_v1}
    \end{minipage}
    \hspace{2mm}
    \begin{minipage}[t]{0.24\linewidth}
        \centering
        \includegraphics[width=0.84\textwidth]{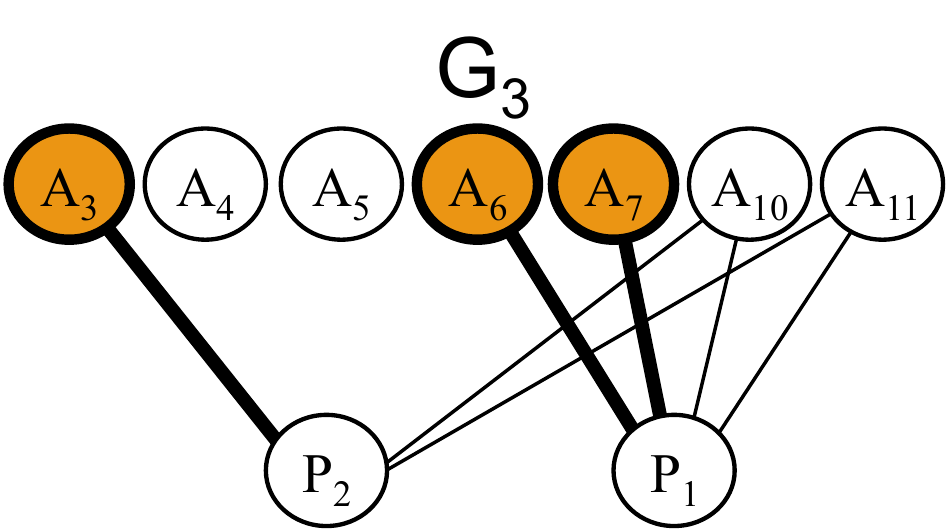}
        \caption{$v_2 = 00100110000$ will be assigned to $P_3$, $v_1 \in \alpha(P_3,G_3)$.}
        \label{fig:E3bipartite_v2}
    \end{minipage}
\end{figure*}

Figure~\ref{fig:bipartite} show the bipartite graphs $G_{1}$, $G_{2}$ and $G_{3}$ for the maximal frequent nodes $v_{1984}$, $v_{248}$ and $v_{499}$ of Example 2. 
As shown in Figures~\ref{fig:bipartite}(a) the bipartite graphs $G_{1}$ for $P_1$ has only the attribute nodes $A_1$, $A_2$, $A_3$, $A_4$, and $A_5$ which belong to $\mathcal{A}_{1984}$ at the top level and the other part is empty because there is no pattern before $P_1$.
The nodes at the top level of $G_2$ (Figures~\ref{fig:bipartite}(b)) are the attributes that belong to $\mathcal{A}_{248}$ and the bottom level has only one node $P_1$ which is the pattern of $v_{1984}$.
Since $A_6$, $A_7$, and $A_8$ belong to $v_{248}$ but do not belong to $v_{1984}$ (i.e., $P_1[6] = P_1[7] = P_1[8] = 0 $), there is an edge between these nodes and $P_1$.
Similarly, the nodes at the top level of $G_3$, Figures~\ref{fig:bipartite}(c), are the attributes that belong to $\mathcal{A}_{499}$; however, the bottom level has two nodes $P_1$ and $P_2$, patterns for $v_{1984}$ and $v_{248}$, respectively.
Since $A_6$, $A_7$, $A_{10}$, and $A_{11}$ belong to $\mathcal{A}_{499}$ but not to $\mathcal{A}_{1984}$ (i.e., $P_1[6] = P_1[7] = P_1[10] = P_1[11] = 0 $), there is an edge between those attributes and $P_1$. Likewise, $A_3$, $A_{10}$, and $A_{11}$ belong to $\mathcal{A}_{499}$ but do not belong to $\mathcal{A}_{248}$ (i.e., $P_2[3] = P_2[10] = P_2[11] = 0 $), there is an edge between those attributes and $P_2$.

Our goal is to assign every frequent node in $\mathcal{U}_{v_i}$ to one and only one partition. We first devise a rule that given an initial pattern $P_j$ and its corresponding bipartite graph $G_j$ allows us to decide if a given frequent node $v_l'$ should belong to a partition with $P_j$ as its representative or not, i.e., $v_l' \in \alpha(P_j,G_j)$. Next, we prove in Theorem~\ref{theo:rule1} that this rule enables us to assign every frequent node to one and only one pattern $P_j$.

\noindent {\bf Rule 1: } {\em 
Given a pattern $P_j$, its bipartite graph $G_j$, and a node $v_{l'}\in COV(P_j)$, node $v_{l'}$ will be assigned to the partition corresponding to $P_j$, i.e., $v_{l'} \in \alpha(P_j,G_j)$, iff for every $P_r$ in the bottom level of $\mathcal{V}(G_j)$, there is at least one attribute $A_k\in \mathcal{A}_{l'}$ in bottom level of $\mathcal{V}(G_j)$ where the edge $\xi_{k,r}$ belongs to $\mathcal{E}(G_j)$}.

In Example 2, consider the pattern $P_3 = 00XXXXX00XX$ and its corresponding bipartite graph $G_3$ in Figure \ref{fig:bipartite}(c). Clearly, $v_1 = 00010110000$ and $v_2 = 00100110000$ are frequent nodes because both are in $COV(P_3)$. However based on rule 1, using the bipartite graph $G_3$, we do not assign $v_1$ to the partition corresponding to $P_3$. In Figure~\ref{fig:E3bipartite_v1}, attributes belonging to $v_1$ are colored. One can observe that, none of the attributes connected to $P_2$ are colored. Thus, node $v_1$ should not be in $P_3$. On the other hand, in Figure~\ref{fig:E3bipartite_v2}, there is at least one attribute connected to $P_1$ and $P_2$, which are colored, i.e., node $v_2 = 00100110000$ is assigned to $P_3$.

\begin{theorem}
\label{theo:rule1}
Given all initial patterns $\{P_j\}$ for each maximal frequent node $v_{j'} = \mathcal{F}_{v_i}[j]$ and the corresponding bipartite graph $G_j$ for every $P_j$, rule 1 guarantees that each frequent node $v_r$ will be assigned to one and only one pattern $P_j$.
\end{theorem}

\begin{proof}
Consider a frequent node $v_r\in\mathcal{U}_{v_i}$.\\
First, since $\mathcal{U}_{v_i} = \underset{\forall v_j\in\mathcal{F}_{v_i}}{\cup}\mathcal{L}_{\mathcal{A}_j}$, there exists at least one node $v_{j'}=\mathcal{F}_{v_i}[j]$ such that $v_r\in COV(P_j)$.
Let $v_{j'}=\mathcal{F}_{v_i}[j]$ be the first nodes in $\mathcal{F}_{v_i}$ where $v_r\in COV(P_j)$, i.e., $\forall v_{l'}=\mathcal{F}_{v_i}[l]$ where $l<j$, $v_r\notin COV(P_l)$. Here we show that $v_r$ is assigned to $P_j$.
Consider any node $v_{l'}=\mathcal{F}_{v_i}[l]$ where $l<j$.
Since, $v_r\in COV(P_j)$, for each attribute $A_k\in \mathcal{A}_r$, $A_k$ belongs to $\mathcal{A}_{j'}$, i.e., $P_j[k]=X$.
On the other hand, since $v_r\notin COV(P_l)$, there exists an attribute $A_k\in \mathcal{A}_k$ where $A_k$ does not belongs to $\mathcal{A}_{l'}$, i.e., $P_l[k]=0$. Thus, for each initial pattern $P_l$ where $l<j$, there exists an edge $\xi_{k,l}\in\mathcal{E}(G_j)$ where $A_k\in \mathcal{A}_r$, and consequently $v_r$ is assigned to $P_j$.\\
The only remaining part of the proof is to show that $v_r$ is assigned to only one pattern $P_j$.
Clearly $v_r$ is not assigned to any pattern $P_{l}$ where $v_r\notin COV(P_l)$.
Again, let $v_{j'}=\mathcal{F}_{v_i}[j]$ be the first nodes in $\mathcal{F}_{v_i}$ where $v_r\in COV(P_j)$.
Since, $v_r\in COV(P_j)$, $\nexists A_k\in \mathcal{A}_r$ where $P_j[k]$ is $0$.
Thus for a node (if any) $v_{l'}=\mathcal{F}_{v_i}[l]$ where $l>j$ and $v_r\in COV(P_l)$, there is no attribute $A_k\in\mathcal{A}_r$ where $P_l[k]=X$ and $P_j[k]=0$. Thus, for any such pattern $P_l$, the assignment $v_r$ to $P_l$ violates Rule 1 ($v_r$ is not assigned to $P_l$).
\end{proof}

In the following, we first show how to efficiently generate the patterns and their corresponding bipartite graphs; then we explain how to use these bipartite graphs to calculate FBC.


\subsection{Efficient generation of bipartite graphs}\label{subsec:patterngen}
Our objective here is to construct the patterns and corresponding bipartite graphs for each node $v_{j'} = \mathcal{F}_{v_i}[j]$. 
The direct construction of the bipartite graphs is straight forward:
$\forall 1\leq l \leq j-1$ let $v_{j'}$ be $\mathcal{F}_{v_i}[j]$ and $v_{l'}$ be $\mathcal{F}_{v_i}[l]$; $\forall k$ where $A_k\in \mathcal{A}_{j'}$ and $A_k\notin \mathcal{A}_{l'}$, add $P_l$ to $\xi_j^k$. This method however is in $O(m|\mathcal{F}_{v_i}|^2)$.
Instead, we propose the following algorithm that is linear to $|\mathcal{F}_{v_i}|$. The key idea is to construct the graph, in an attribute-wise manner.

For every attribute $A_k\in \mathcal{A}_{j'}$, we maintain a list ($\xi_j^k$) for all $v_{j'} = \mathcal{F}_{v_i}[j]$ in order to store (the edges of) $G_j$;
for each edge $\xi_{k,l}$ from $A_k$ to $P_l$ in $G_j$, we add the index of $P_l$ to the corresponding list of $A_k$. $\xi_j^k$ represents the adjacent nodes to $A_k$
in bipartite graph $G_j$.
For example, in Figure~\ref{fig:bipartite}(c), which is the corresponding bipartite graph for $P_3$ ($j = 3$ and $k \in \{3,4,5,6,7,10,11\}$), $\xi^3_3=\{ P_2\}$, $\xi_3^6=\{ P_1\}$, $\xi_3^7=\{ P_1\}$, $\xi_3^{10}=\{ P_1, P_2\}$, and $\xi_3^{11}=\{ P_1, P_2\}$, while $\xi_3^4$ and $\xi_3^5$ are empty sets.

For each attribute $A_k$, let us define the inverse index $\delta_k$
as a binary vector of size $|\mathcal{F}_{v_i}|$; 
for every node $v_{j'}=\mathcal{F}_{v_i}[j]$ if $A_k\in \mathcal{A}_{j'}$ then $\delta_k[j]$ is $1$ and $0$ otherwise. In Example 2, for attribute $A_3$, $\delta_6 =[1,0,1]$ because $A_3 \in \mathcal{A}_{1984}$ ($v_{1984} =\mathcal{F}_{v_i}[1]$) and $A_6 \in \mathcal{A}_{499}$ ($v_{499} =\mathcal{F}_{v_i}[3]$) but $A_3 \notin \mathcal{A}_{248}$ ($v_{248} =\mathcal{F}_{v_i}[2]$).
Similarly, for attribute $A_6$, $\delta_6 =[0,1,1]$ because $A_6 \notin \mathcal{A}_{1984}$ but $A_6 \in \mathcal{A}_{248}$ and $A_6 \in \mathcal{A}_{499}$.
One observation is that for an index $j_1 < j_2$ where $\delta_k[j_1]=1$ and ${\delta_k[j_2]}=1$, $\xi_{j_1}^k$ is the subset of $\xi_{j_2}^k$.
$\xi_{j_1}^k$ is the set of indices $P_l$ for $1\leq l < j_1$ where $A_k\notin \mathcal{A}_{l'}$ where $v_{l'}=\mathcal{F}_{v_i}[l]$.
Since $j_1 < j_2$, if $P_l\in \xi_{j_1}^k$ then $P_l\in \xi_{j_2}^k$.
Specifically, considering two indices $j_1 < j_2$ where $\delta_k[j_1]=1$, ${\delta_k[j_2]}=1$, and $\forall j_1<l<j_2$ $\delta_j[l]=0$: $$\xi^k_{j_2} = \,\xi^k_{j_1} \cup \{ P_l | j_1<l<j_2\}$$

In Example 2, for attribute $A_3$ where $\delta_3 =[1,0,1]$, $\xi_1^3 = \emptyset$ (in $G_1$ the bottom level is empty), $\xi_3^3 = \{P_2\}$ (adjacent nodes of $A_3$ in Figure~\ref{fig:bipartite}(c)). We observe that $\xi_1^3 \subseteq \xi_3^3$. For attribute $A_6$ where $\delta_6 =[0,1,1]$, $\xi_2^6 = \{P_1\}$ (adjacent nodes of $A_6$ in Figure~\ref{fig:bipartite}(b)), $\xi_3^6 = \{P_1\}$ (adjacent nodes of $A_6$ in Figure~\ref{fig:bipartite}(c)) and $\xi_2^6 \subseteq \xi_3^6$. 
For example, for a node $v_i$, suppose $\mathcal{F}_{v_i}= \{ v_{i_1},v_{i_2}, \dots , v_{i_8} \}$. For an attribute $A_k$, let $\delta_k$ be $[0,0,1,0,0,1,1,0]$; in this example, $A_k$ only belongs $\mathcal{A}_{i_3}$, $\mathcal{A}_{i_6}$, and $\mathcal{A}_{i_7}$.
Since $A_k$ does not belong to $\mathcal{A}_{i_1}$ and $\mathcal{A}_{i_2}$, $\xi_k^3 = \{ P_1, P_2 \}$. In fact, $\delta_k[1]$, and $\delta_k[2]$ are the zeros before $\delta_k[3] = 1$. Similarly, $\xi_k^6 = \{ P_1, P_2 , P_4, P_5\}$ ($\delta_k[1]$, $\delta_k[2]$, $\delta_k[4]$, and $\delta_k[5]$ are the zeros before the $\delta_k[6] = 1$). Also, $\xi_k^3\subset \xi_k^6$ and more precisely $\xi_k^6 = \xi_k^3 \cup \{ P_4, P_5\}$.

We use this observation to design Algorithm~\ref{alg:FBC_bipartite} to generate the corresponding patterns and bipartite graphs
that partition $\mathcal{U}_{v_i}$.
After transforming each node $v_{j'}=\mathcal{F}_{v_i}[j]$ to $P_j$ (by changing $1$'s in $\mathcal{B}(v_{j'})$ to $X$) and generating $\delta_k$ (lines 3-7), for every attribute $A_k$ the algorithm iterates over $\delta_k$ and gradually constructs the bipartite graphs. Specifically, for every attribute $A_k$, it iterates over the inverse index $\delta_k$; if $\delta_k[j]=0$, then it adds $P_j$ to the temporary set (temp) and if $\delta_k[j]=1$, it adds temp to $\xi_j^k$ (lines 10-17).
In this way the algorithm conducts only one pass over $\delta_k$ (with size $|\mathcal{F}_{v_i}|$) for each attribute $A_k$. Thus, the computational complexity of Algorithm~\ref{alg:FBC_bipartite} is in $O(m|\mathcal{F}_{v_i}|)$.

\begin{algorithm}[!h]
\caption{{\bf ConstructBipartiteGraphs}\\
		 {\bf Input:} node $v_i$ and maximal frequent nodes $\mathcal{F}_{v_i}$ 
		}
\begin{algorithmic}[1]
\label{alg:FBC_bipartite}
\STATE $P = [~]$ , $\xi = [~][~]$
\FOR {$j=1$ to $|\mathcal{F}_{v_i}|$}
	\STATE $v_{j'}=\mathcal{F}_{v_i}[j]$
	\FOR{$k=1$ to $\leq |\mathcal{A}_i|$}
		\STATE $P_j[k] = X$ {\bf if} $A_k\in \mathcal{A}_{j'}$ {\bf otherwise} $0$
		\STATE $\delta_k[j] = 1$ {\bf if} $A_k\in \mathcal{A}_{j'}$ {\bf otherwise} $0$
	\ENDFOR
	\STATE $P[j] = P_j$ //Add $P_j$ to the patterns 
\ENDFOR
\FOR {$k=1$ to $\leq |\mathcal{A}_i|$}
	\STATE temp = $\{ \}$
	\FOR{$j = 1$ to $|\mathcal{F}_{v_i}|$}
		\STATE {\bf if} $\delta_k[j]=0$ {\bf then} add $P_j$ to temp
		\STATE {\bf else} $\xi_j^k = $ temp
		\STATE \hspace{0.2in} $\xi[j].append (\xi_j^k)$
	\ENDFOR
\ENDFOR
\STATE {\bf return} ($P$,$\xi$)
\end{algorithmic}
\end{algorithm}

\subsection{Counting}
So far, we partitioned $\mathcal{U}_{v_i}$, generating the patterns and their corresponding bipartite graphs, while enforcing Rule 1.
Still, we need to identify the number of nodes assigned to each pattern $P_j$, while enforcing Rule 1 (i.e., determine $|\alpha(P_j,G_j)|$) in order to calculate the FBC.
As previously discussed, the number of nodes in $|COV(P_j)|$ is $2^{k_x(P_j)}$, where $k_x(P_j)$ is the number of $X$'s in $P_j$. However, enforcing Rule 1, 
the number of nodes assigned to partition $P_j$, $|\alpha(P_j,G_j)|$, is less than $|COV(P_j)|$.

Since Rule 1 partitions $\mathcal{U}_{v_i}$ to disjoint sets of $\alpha(P_j,G_j)$, $\forall v_j\in \mathcal{F}_{v_i}$, FBC$(\mathcal{B}(v_i))$ can be calculated by summing the number of nodes assigned to each pattern $P_j$; i.e.,
\begin{align}
\mbox{FBC}(\mathcal{B}(v_i)) = \sum\limits_{v_j\in \mathcal{F}_{v_i}}|\alpha(P_j,G_j)|
\end{align}
Thus, in Example 2, using the bipartite graphs in Figure~\ref{fig:bipartite} we can calculate FBC$(\mathcal{B}(v_i))$ as $|\alpha(P_1,G_1)|$ + $|\alpha(P_2,G_2)|$ + $|\alpha(P_3,G_3)|$.

We propose the recursive Algorithm~\ref{alg:patterncount} to calculate $|\alpha(P_j,G_j)|$. 
The algorithm keeps simplifying $G_j$ removing the edges (as well as the edges in the bottom part that are not connected to any nodes), until it reaches a simple (boundary) case
with no edges in the graph such that $|\alpha(P_j,G_j)|$ can be computed directly.
For example the bipartite graph $G_1$ for pattern $P_1$ in Figure~\ref{fig:bipartite}(a) has no edges thus Algorithm~\ref{alg:patterncount} returns $|\alpha(P_1,G_1)| = 2^5 = 32$.
Figure~\ref{fig:bipartite}(b) presents a more complex example ($G_2$) in which there is one node in the bottom part of the graph having three edges. In this example, out of the $2^3$ combinations of the three attributes connected to $P_1$, one (where all of them are $0$) violates Rule 1. As a result, $|\alpha(P_1,G_1)|$ is $2^3-1$ 
times the number of valid combinations of the rest of the attributes (i.e., $A_4$ and $A_5$). Thus, $|\alpha(P_1,G_1)| = (2^3-1)\times 2^2 = 28$.
In more general cases, each attribute node in the upper part may be connected to several nodes in the bottom part.
In these cases, in order to simplify $G_j$ faster, the algorithm considers a set $S$ of attributes connected to the same set of nodes
, while maximizing the number of nodes that are connected to $S$.
In Figure~\ref{fig:bipartite}(c) (showing $G_3$) attributes $A_{10}$ and $A_{11}$ are connected to both $P_1$ and $P_2$, i.e., $|\xi_3^{10}|$ and $|\xi_3^{11}|$ are $|\{P_1,P_2\}|=2$ while for other nodes it is less than $2$.
If at least one of the attributes in $S$ is not $0$,
for any node $P_l$ connected to $S$, all the edges that are connected to $P_l$ can be removed from $\mathcal{E}(G_j)$.
In Figure~\ref{fig:bipartite}(c), if either of $A_{10}$ or $A_{11}$ is not $0$, all the edges connected to $P_1$ and $P_2$ can be removed.
The number of such combinations is $2^{|S|}-1$. Thus, the algorithm conducts the updates in $G_j$, recursively counting the number of 
remaining combinations, multiplying the result with $2^{|S|}-1$.
The remaining case is when all the attributes in $S$ are $0$. 
This case is only allowed if for all $P_l$ connected to $S$, there exists an index $q'\notin S$ where $P_l\in \xi_j^{q'}$.
In this case, for each $P_l$ connected to $S$, one of the attributes $A_{q'}$ where $q'\notin S$ should be non-zero in accordance to Rule 1. 
Thus, $\forall q\in S$, the algorithm sets $P_j[q]$ to zero and recursively computes the number of such nodes in $\alpha(P_j,G_j)$.
Following Algorithm~\ref{alg:patterncount} to calculate the $|\alpha(P_3,G_3)|$ in Figure~\ref{fig:bipartite}, we have $(2^2 -1)2^5 + (2^2-1)(2^1-1)2^2 =96$.

\begin{algorithm}[!ht]
\caption{{\bf CountPatterns}\\
         {\bf Input:} Node $v_i$, pattern $P_j$, and edges $\xi_j$ of the bipartite graph $G_j$\\
          {\bf Output:} $|\alpha(P_j,G_j)|$
        }
\begin{algorithmic}[1]
\label{alg:patterncount}
\STATE {\bf if} $|\xi_j|=0$ {\bf then return} $2^{k_x(P_j)}$ {\it \small{// $k_x(P_j)$ is the number $X$'s in $P_j$}}
\STATE $q_{max} = \underset{1\leq q \leq |\mathcal{A}_i|}{\mbox{argmax}} ( |\xi_j^q|)$
\STATE $S = \{q |  \xi_j^q = \xi_j^{q_{max}}\}$
\STATE $\xi$ = remove $S$ from $\xi_j$; $P$ = $P_j$
\STATE {\bf for all} $q\in S$ {\bf do} $P[q]=0$
\IF{$\forall P_l \in \xi_j^{q_{max}}:$ $\exists q'\notin S$ where $P_l\in \xi_j^{q'}$}
    \STATE count = {\bf CountPatterns}($v_i$, $P$, $E$)
\ELSE
    \STATE count $ = 0$
\ENDIF
\STATE remove all the instances of $P_l\in \xi_j^{q_{max}}$ from $\xi$
\STATE {\bf return} count + $(2^{|S|}-1)\times ${\bf CountPatterns}($v_i$, $P$, $E$)
\end{algorithmic}
\end{algorithm}

Utilizing Algorithms~\ref{alg:FBC_bipartite} and~\ref{alg:patterncount}, Algorithm~\ref{alg:FBC} presents our final algorithm to compute FBC($\mathcal{B}(v_i)$). 
In Example 2, Algorithm~\ref{alg:FBC} returns FBC$(\mathcal{B}(v_i))$ as $|\alpha(P_1,G_1)|$ + $|\alpha(P_2,G_2)|$ + $|\alpha(P_3,G_3)| = 32 +28+96 = 156$.

\begin{algorithm}[!ht]
\caption{{\bf FBC}\\
         {\bf Input:} Node $v_i$ and maximal frequent nodes $\mathcal{F}_{v_i}$ 
        }
\begin{algorithmic}[1]
\label{alg:FBC}
\STATE sort $\mathcal{F}_{v_i}$ descendingly based on $\ell(v_i)$
\STATE $P,\xi$ = {\bf ConstructBipartiteGraphs}($v_i$, $\mathcal{F}_{v_i}$)
\STATE count = 0
\FOR{$j=1$ to $|\mathcal{F}_{v_i}|$}
\STATE count = count + {\bf CountPatterns}($v_i$, $P[j]$, $\xi[j]$)
\ENDFOR
\STATE {\bf return} count
\end{algorithmic}
\end{algorithm}

\begin{theorem}
The time complexity of Algorithm~\ref{alg:FBC} is $O\big((m+|\mathcal{F}_{v_i}|)\times\min(\mbox{FBC}(\mathcal{B}(v_i)), |\mathcal{F}_{v_i}|^2) \big)$.
\end{theorem}

\begin{proof}
According to \S~\ref{subsec:patterngen}, line 2 is in $O(m|\mathcal{F}_{v_i}|)$. This is because Algorithm~\ref{alg:FBC_bipartite} conducts 
only one pass per attribute $A_k$ over $\delta_k$ (with size $|\mathcal{F}_{v_i}|$) and amortizes the graph construction cost.
Lines 3 to 6 of Algorithm~\ref{alg:FBC} call the recursive Algorithm~\ref{alg:patterncount} to count the total number of nodes assigned to the patterns.
For each initial pattern $P_j$ and its bipartite graph $G_j$ there exists at least one frequent node $v_r$ where $v_r$ is {\em only} in $COV(P_j)$ (i.e., for any other initial pattern $P_l$, $v_r$ does not belong to $COV(P_l)$), otherwise the node $v_{j'}=\mathcal{F}_{v_i}[j]$ is not maximal frequent.
Moreover, at every recursion of Algorithm~\ref{alg:patterncount}, at least one node that is assigned to $P_j$ is counted; that can be observed from lines 1 and 11 in Algorithm~\ref{alg:patterncount}.
Thus, Algorithm~\ref{alg:patterncount} is at most FBC$(v_i)$ times is called during the recursions.
On the other hand, at each recursion of Algorithm~\ref{alg:patterncount}, at least one of the nodes $P_l$ in $\mathcal{V}(G_j)$ is removed. As a result, for each initial pattern $P_j$, the algorithm is called at most $|\mathcal{F}_{v_i}|$ times; this provides another upper bound on the number of calls that Algorithm~\ref{alg:patterncount} is called during the recursion as $|\mathcal{F}_{v_i}|^2$.
At every recursion of the algorithm, finding the set of attributes with maximum number of edges is in $O(m)$. Also updating the edges and nodes of $G_j$ is in $O(m+|\mathcal{F}_{v_i}|)$.
Therefore, the worst-case complexity of Algorithm~\ref{alg:FBC} is $O\big((m+|\mathcal{F}_{v_i}|)\times\min(\mbox{FBC}(\mathcal{B}( v_i)), |\mathcal{F}_{v_i}|^2) \big)$.
\end{proof}
 \section{Discussions}\label{sec:discussion}

\subsection{Generalization to ordinal attributes}\label{subsec:disc-categorical}

The algorithms proposed in this paper are built around the notion of a lattice of attribute combinations;
extending this notion to ordinal attributes is key.
Let Dom($A_k$) be the cardinality of an attribute $A_k$.
We define the vector representative $\nu_i$ for every attribute-value combination $\mathcal{C}_i$ as the vector of size $m$, where $\forall 1\leq k \leq m$, $\nu[k]\in$ Dom$(A_k)$ is the value of the attribute $A_k$ in $\mathcal{C}_i$.
$\nu_i$ dominates $\nu_j\neq \nu_i$, {\it iff} $\forall A_k\in\mathcal{A}$, $\nu_j[k]\leq \nu_i[k]$.

Generalizing from binary attributes to ordinal attributes, the DAG (directed acyclic graph) of
attribute-value combinations is defined as follows:
Given the vector representative $\nu_i$, the DAG of $\nu_i$ is defined as $\mathfrak{D}_{\nu_i} = (V,E)$,
where $V(\mathfrak{D}_{\nu_i})$, corresponds to the set of all vector representatives dominated by $\nu_i$;
thus for all $\nu_j$ dominated by $\nu_i$, there exists a one to one mapping between each $v_j\in V(\mathfrak{D}_{\nu_i})$ and each $\nu_j$.
Let maxD be $\max_{k=1}^m($Dom$(A_k))$.
For consistency, for each node $v_j$ in $V(\mathfrak{D}_{\nu_i})$, the index $j$ is the decimal value of the $\nu_j$, in base-maxD numeral system.
Each node $v_j\in V(\mathfrak{D}_{\nu_i})$ is connected to all nodes $v_l\in V(\mathfrak{D}_{\nu_i})$ as its parents (resp. children), for $m-1$ attributes $A_{k'}$, $\nu_j[k']=\nu_l[k']$ and for the remaining attribute $A_k$, $\nu_l[k]=\nu_j[k]+1$ (resp. $\nu_l[k]=\nu_j[k]-1$).
The level of each node $v_j\in V(\mathfrak{D}_{\nu_i})$ is the sum  of its attribute values, i.e., $\ell(v_j)=\sum_{k=1}^m \nu_j[k]$.
In addition, every node $v_j$ is associated with a cost defined as $cost(v_j) = \sum_{k=1}^m cost(A_k,\nu_j[k])$.

As an example let $\mathcal{A}$ be $\{A_1,A_2\}$ with domains Dom($A_1$)=3 and Dom($A_2$)=3.
Consider the vector representative $\nu_{8}=\langle 2,2\rangle$; Figure~\ref{fig:d2} shows $\mathfrak{D}_{\nu_8}$.
In this example node $v_7$ is highlighted; the vector of $v_7$, $\nu_7 = \langle 2,1\rangle$, represents the attribute-value combination $\{ A_1=2, A_2=1\}$, and $\ell (v_7)=3$.

\begin{figure}[!t]
    \begin{minipage}[t]{0.51\linewidth}
        \centering
        \includegraphics[width=1.1\textwidth]{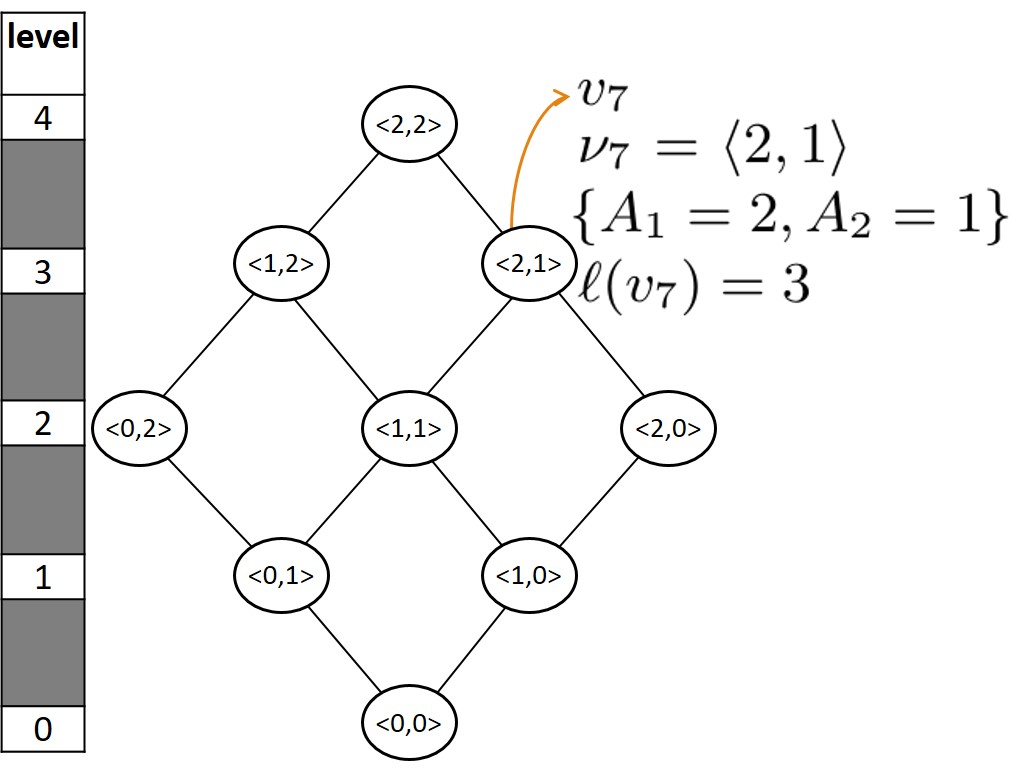}
        \caption{Illustration of $\mathfrak{D}_{\nu_8}$} 
		\label{fig:d2}
    \end{minipage}
    \begin{minipage}[t]{0.47\linewidth}
        \centering
        \vspace{-31mm}
        \includegraphics[width=0.7\textwidth]{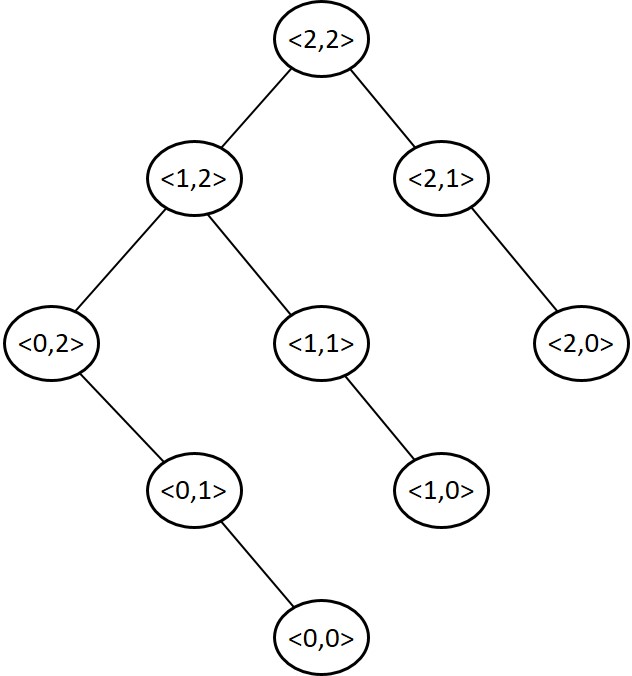}
        \caption{Tree construction for Figure~\ref{fig:d2}} 
		\label{fig:d2t}
    \end{minipage}
\end{figure}


The definition of maximal affordable nodes in the DAG is the same as Definition~\ref{def:maximalaffordable}, i.e., the nodes that are affordable and none of their parents are affordable are maximal affordable.
The goal of the general solution, as explained in \S~\ref{sec:exact}, is to efficiently idemtify the set of
maximal affordable nodes. Using the following transformation of the DAG to a tree data structure,
{\bf G-GMFA} is generalized for ordinal attributes: 
for each node $v_i$, the parent $v_j$ with the minimum index $j$ among its parents is responsible for generating it. 
Figure~\ref{fig:d2t} presents the tree transformation for $\mathfrak{D}_{\nu_8}$ in Figure~\ref{fig:d2}.
Also pre-ordering the attributes removes the need to check if a generated node has an affordable parent and the nodes that have an affordable parent never get generated.

Similarly, $\mathcal{F}_{v_i}$ is defined as the set of frequent nodes $v_j\in \mathfrak{D}_{\nu_i}$ that do not have any frequent parents; having the set of maximal frequent nodes $\mathcal{F}_{v_i}$, for each node $v_j\in\mathcal{F}_{v_i}$, we define the bipartite graph similar to \S~\ref{sec:gain}; the only difference is that for each edge $\xi_{k,l}$ in $\mathcal{E}(G_j)$ we also maintain a value $w_{k,l}$ that presents the maximum value until which
$A_k$ can freely reduce, without considering the values of other nodes connected to $P_l$.
Similarly, the size of $\mathcal{U}_{v_i}$ is computed iteratively by reducing the size of $G_j$.

\subsection{Gain function design with additional information}\label{subsec:dis-userpref}
In \S~\ref{sec:exact}, we proposed an efficient general framework that works for any 
arbitrary monotonic gain function.  We demonstrate that numerous 
application specific gain functions can 
be modeled using the proposed framework.
In \S~\ref{sec:gain}, we studied FBC, as a gain function applicable in cases such as third party services that may only have access
to dataset tuples. However, the availability of other information such as user preferences enables an understanding of demand towards attribute combinations. For example, the higher the number of queries requesting an attribute combination that is underrepresented in existing competitive rentals in the market or the higher ratings for properties with a specific attribute combination, the higher the potential
gain. In the following, we discuss a few examples of gain functions
based on the availability of extra information and the way those can be modeled
in our framework.


\vspace{0.05in}
\noindent{\bf Existence of the users' feedback:}
User feedback provided in forms such as online reviews, user-specified tags, and ratings are direct ways of
understanding user preference.
For simplicity, let us assume user feedback for each tuple is provided as a number showing how desirable the tuple is.
Consider the $n$ by $m$ matrix $D$ in which the rows are the tuples and the columns are the attribute, as well as the $1$ by $n$ vector $R$ that shows the user feedback for each tuple. 
Multiplication of $R$ and $D$ as $R' = R \times A$ generate a feedback vector for the attributes. Then, the gain for the attribute combination $\mathcal{A}_i\subseteq \mathcal{A}$ is calculated as:
\begin{align}\label{eq:gain_feedback}
gain(\mathcal{A}_i) = \sum\limits_{\forall A_j\in \mathcal{A}_i} R'[j]
\end{align}
This can simply be applied into the framework by using Equation~\ref{eq:gain_feedback} in line~\ref{line:G-Gmfa-gain} of Algorithm~\ref{alg:general} in order to compute $gain(\mathcal{A}_t\cup \mathcal{A}_v)$.

\vspace{0.05in}
\noindent{\bf Existence of query logs:}
Query logs (workloads) are popular for capturing the user demand~\cite{miah2009}.
A way to design the gain function using a query log is to count the number of queries that return a tuple with that attribute combination.
This way of utilizing the workload, however, does not consider the availability of other
information such as existing tuples in the database. Especially the tuples in the database contain the {\em supply} information, which is missing from the workload as defined above.
This necessitates a practical gain function based on the combination of demand and supply, captured by the workload and by the dataset tuples, accordingly. 
Therefore,
besides more complex ways of modeling it,
the gain of a combination $\mathcal{A}_i$ can be considered as the portion of queries in the workload that are the subset of the combination (demand) divided by the portion of tuples in $\mathcal{D}$ that are the superset of it (supply). Formally:
\begin{align}\label{eq::gain-comb}
gain(\mathcal{A}_i) = \frac{n \times |\{q\in \mathcal{W} | q\subseteq \mathcal{A}_i \}|}{|\mathcal{W}| \times |Query(\mathcal{A}_i, \mathcal{D})|}
\end{align}
Now, at line~\ref{line:G-Gmfa-gain} of Algorithm~\ref{alg:general}, $gain(\mathcal{A}_t\cup \mathcal{A}_v)$ is computed using Equation~\ref{eq::gain-comb}.
\section{Experimental Evaluation}\label{sec:exp}
We now turn our attention to the experimental evaluation of the proposed algorithms.
In this section, after providing the experimental setup  in \S~\ref{subsec:exp-setup}, we evaluate the performance of our proposed algorithms over the real data collected from AirBnB  in \S~\ref{subsec:expresults} . We also provide a case study in \S~\ref{subsec:casestudy} to illustrate the practicality of the approaches. Both experimental results and case study shows the efficiency and practicality our algorithms.
\subsection{Experimental Setup}\label{subsec:exp-setup}
\vspace{0.05 in}
\noindent {\bf Hardware and Platform:} 
All the experiments were performed on a Core-I7 machine running Ubuntu 14.04 with 8 GB of RAM.
The algorithms described in \S~\ref{sec:exact} and \S~\ref{sec:gain} were implemented in Python.

\vspace{0.05 in}
\noindent {\bf Dataset:} 
The experiments were conducted over real-world data collected from \emph{AirBnB}, a travel 
peer to peer marketplace.
We collected the information of approximately 2 million \emph{real} properties around the world, shared on this website. AirBnB has a total number of 41 attributes for each property. Among all the attributes, 36 of them are boolean attributes, such as \texttt{TV}, \texttt{Internet}, \texttt{Washer}, and \texttt{Dryer}, while 5 are ordinal attributes, such as \emph{Number of Bedrooms} and \emph{Number of Beds}.
We identified 26 (boolean) flexible attributes and for practical purposes we estimated their costs for a one year period. Notice that
these costs are provided purely to facilitate the experiments; other values could be chosen and would not affect the relative performance and
conclusions in what follows.
In our estimate for attribute $cost[.]$, one attribute (\texttt{Safety card}) is less than $\$ 10$, nine attributes (eg. \texttt{Iron}) are between $\$ 10$ and $\$ 100$, fourteen (eg. \texttt{TV}) are between $\$ 100$ and $\$ 1000$, and two (eg. \texttt{Pool}) are more than $\$ 1000$.

\vspace{0.05 in}
\noindent {\bf Algorithms Evaluated:}
We evaluated the performance of the three algorithms introduced namely: {\bf B-GMFA} (Baseline GMFA), 
{\bf I-GMFA} (Improved GMFA) and {\bf G-GMFA} (General GMFA). According to \S~\ref{sec:exact}, {\bf B-GMFA} does not consider pruning 
the sublattices of the maximal affordable nodes and examines all the nodes in $\mathcal{L}_\mathcal{A}$.
In addition, for \S~\ref{sec:gain}, the performance of algorithms {\bf A-FBC} (Apriori-FBC) and {\bf FBC} is evaluated. 

\vspace{0.05 in}
\noindent {\bf Performance Measures:} 
Since the proposed algorithms identify the optimal solution, we focus our evaluations on the {\em execution time}. 
Assessing gain is conducted during the execution of GMFA algorithms; thus when comparing the performances of \S~\ref{sec:exact} algorithms, we separate
the computation time to assess gain, from the total time.
Similarly, when evaluating the algorithm in \S~\ref{sec:gain}, we evaluate instances of {\bf G-GMFA} applying each of the three 
FBC algorithms to assess gain and only consider the time required to run each FBC algorithm and compute the gain.
Algorithm~\ref{alg:FBC} ({\bf FBC}) accept the set of maximal frequent nodes ($\mathcal{F}_{v_i}$) as input.
The set of maximal frequent nodes is computed once \footnote{\small{Please note that $\forall v_i\in\mathcal{L}_{\mathcal{A}}$, $\mathcal{F}_{v_i}$ can be computed directly from $\mathcal{F}$ by projecting $\mathcal{F}$ on $v_i$ and removing the dominated nodes from the projection.}} 
and utilized across all the queries.
For fairness, we consider the time to compute $\mathcal{F}_{v_i}$ as part of the gain function computation time.

\vspace{0.05 in}
\noindent {\bf Default Values:} While varying one of the variables in each experiment, we use the following default values for the other ones.
$n$ (number of tuples): 200,000; $m$ (number of attributes): 15; $B$ (budget): \$2000; $\tau$ (frequency threshold): 0.1; $\mathcal{A}_t=\{\}$.

\begin{figure*}[!ht]    
    \begin{minipage}[t]{0.23\linewidth}
        \includegraphics[scale=0.24]{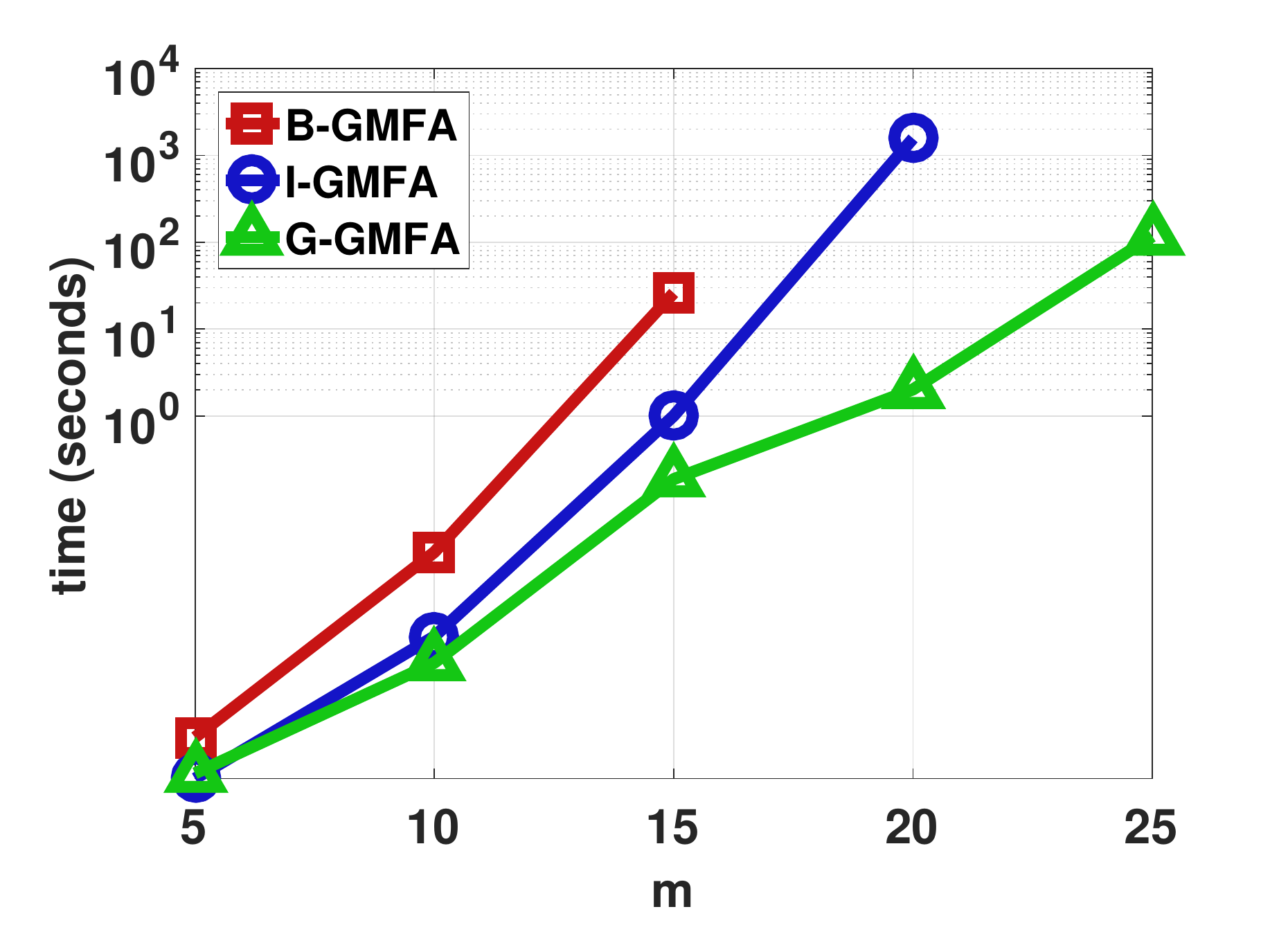}
        \caption{Impact of varying $m$ on GMFA algorithms}
        \label{fig:vm1}
    \end{minipage}
    \hspace{1mm}
    \begin{minipage}[t]{0.23\linewidth}
        \includegraphics[scale=0.24]{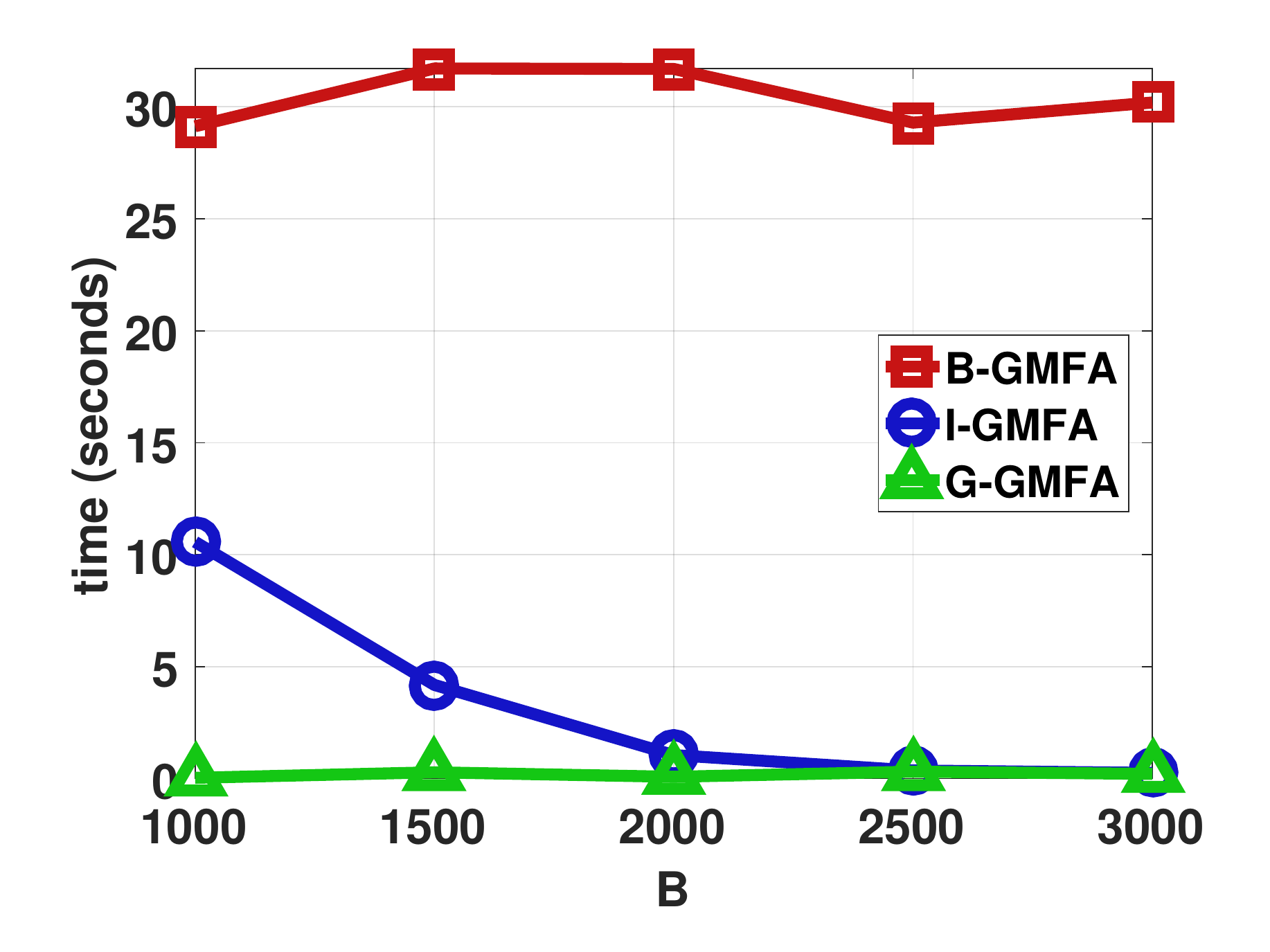}
        \caption{Impact of varying $B$ on GMFA algorithms}
        \label{fig:vb1}
    \end{minipage}
    \hspace{1mm}
    \begin{minipage}[t]{0.23\linewidth}
        \includegraphics[scale=0.24]{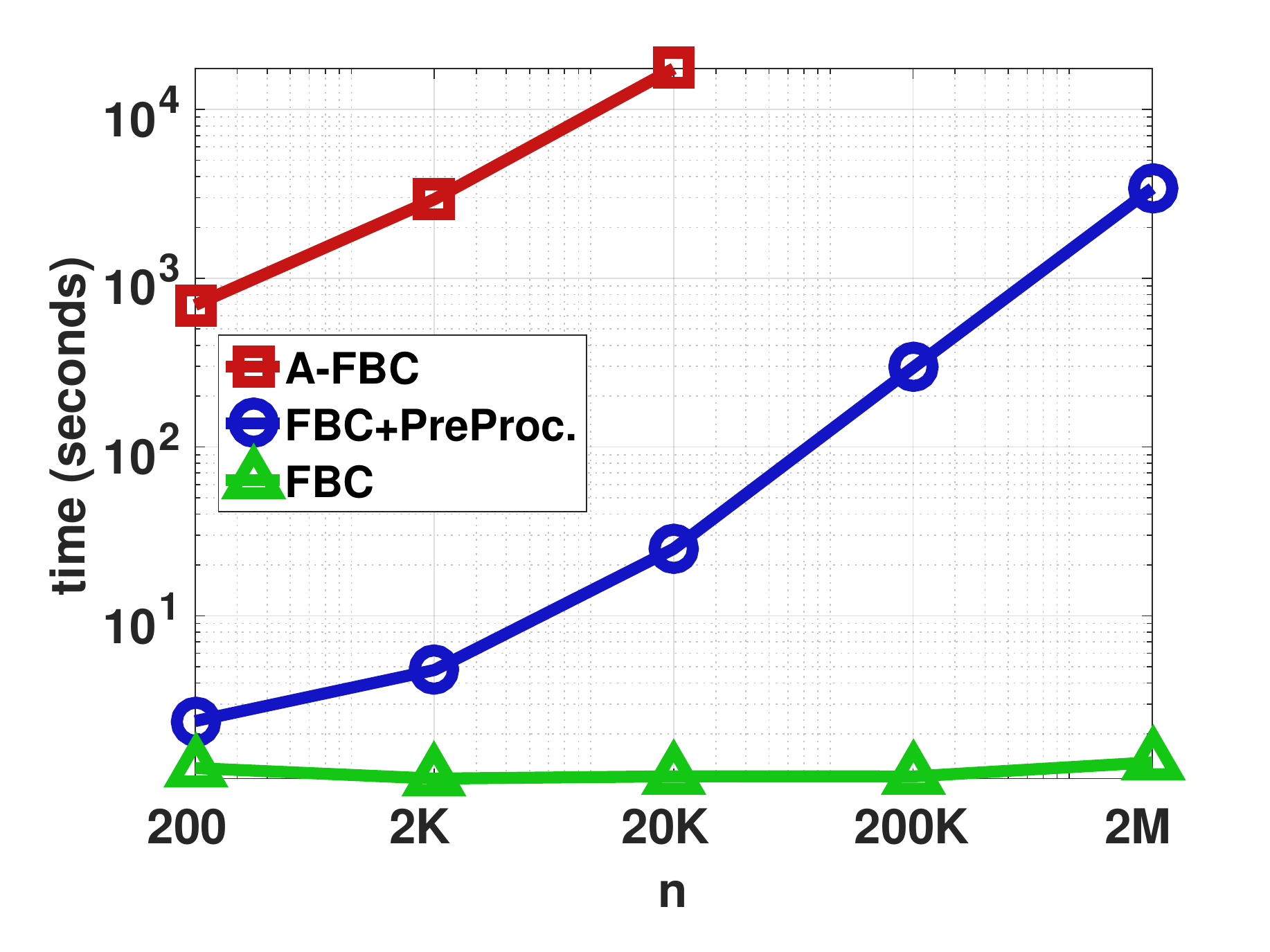}
        \caption{Impact of varying $n$ on FBC algorithms}
        \label{fig:vn1}
    \end{minipage}
    \hspace{1mm}
    \begin{minipage}[t]{0.23\linewidth}
        \includegraphics[scale=0.24]{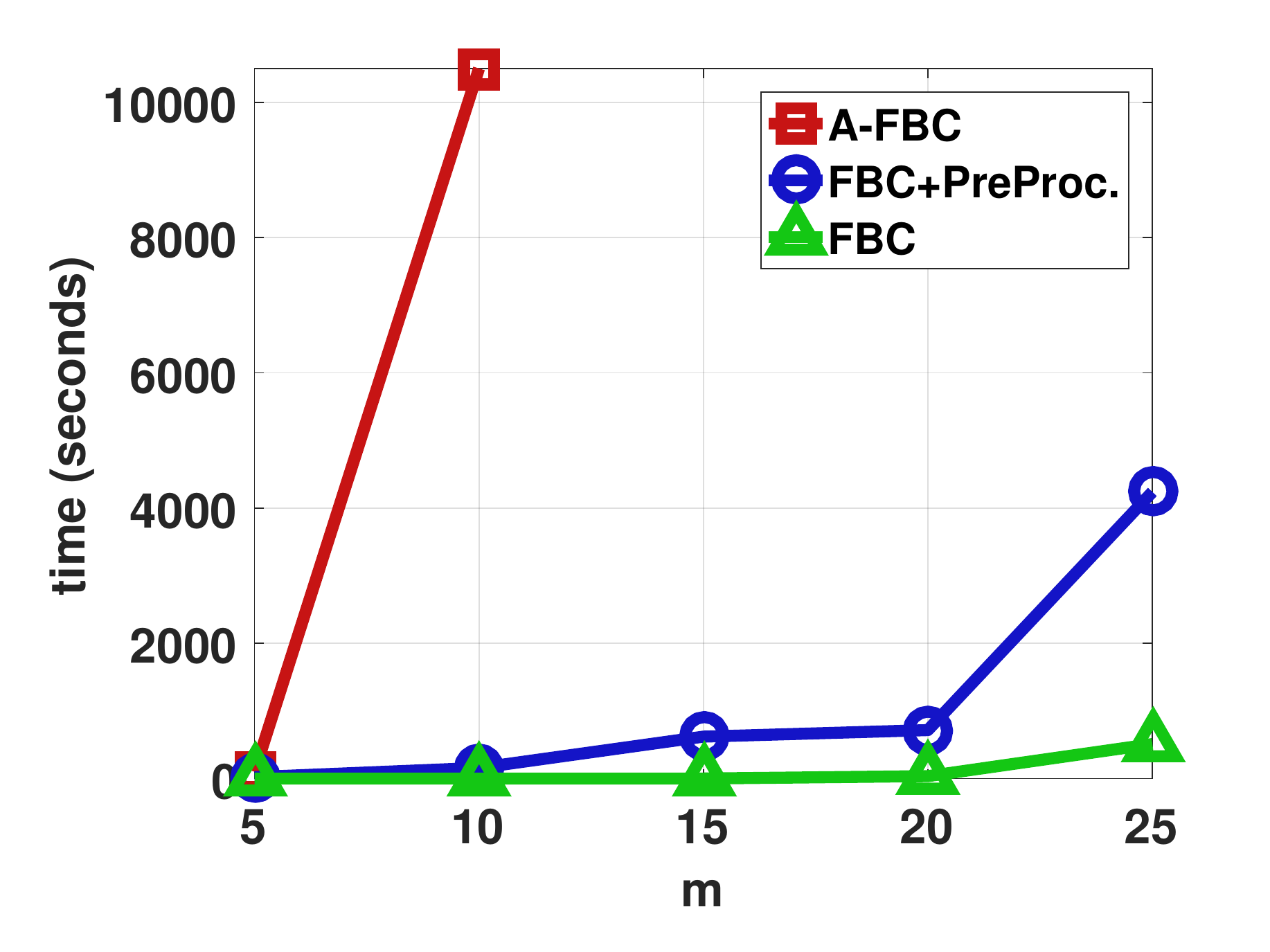}
        \caption{Impact of varying $m$ on FBC algorithms}
        \label{fig:vm2}
    \end{minipage}
    
    \hspace{7mm}
    \begin{minipage}[t]{0.3\linewidth}
        \centering
        \vspace{3mm}
        \includegraphics[scale=0.24]{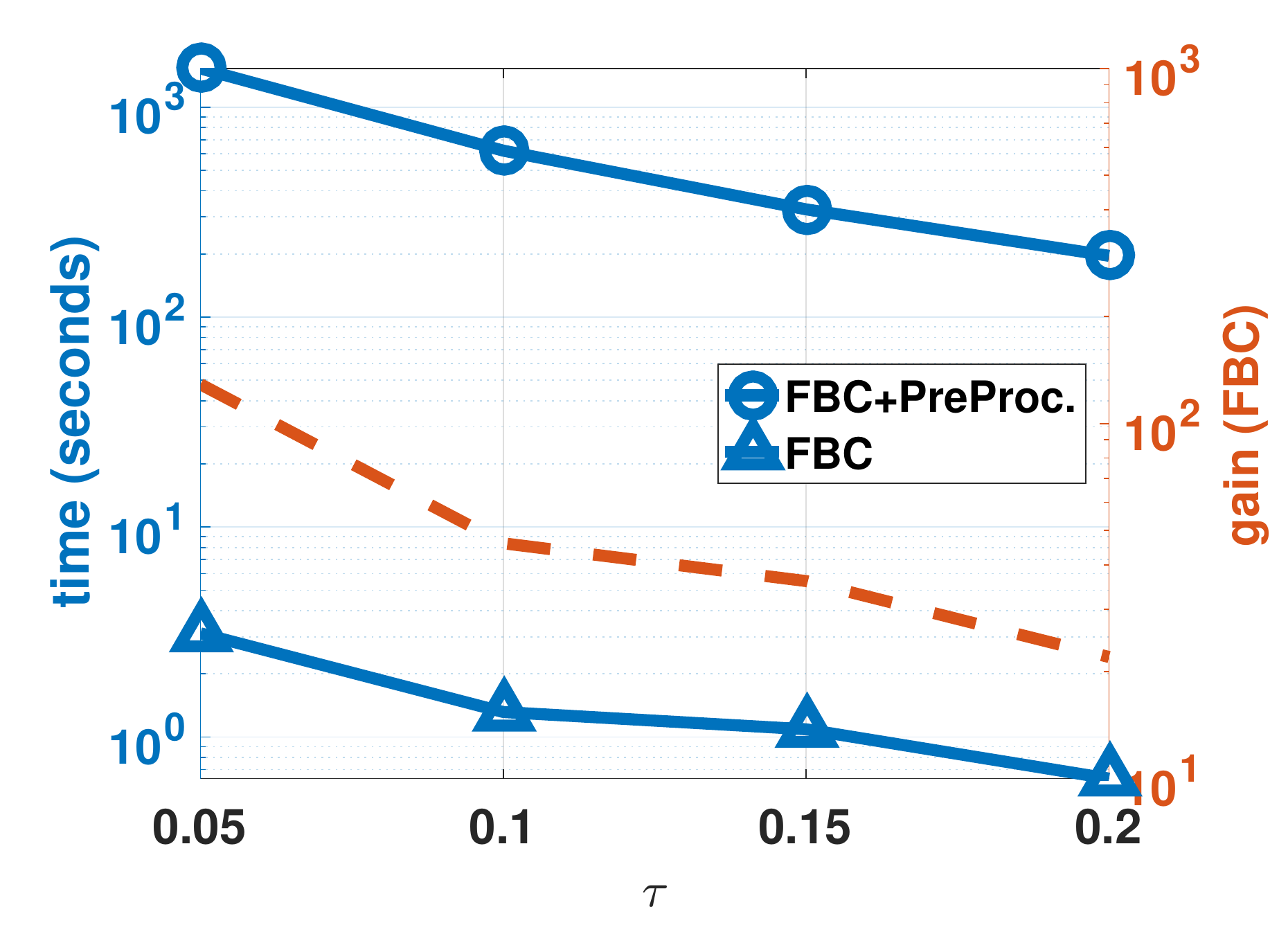}
        \caption{Impact of varying $\tau$ on FBC algorithms}
        \label{fig:vt}
    \end{minipage}
    \hspace{1mm}
    \begin{minipage}[t]{0.3\linewidth}
        \centering
        \vspace{3mm}
        \includegraphics[scale=0.24]{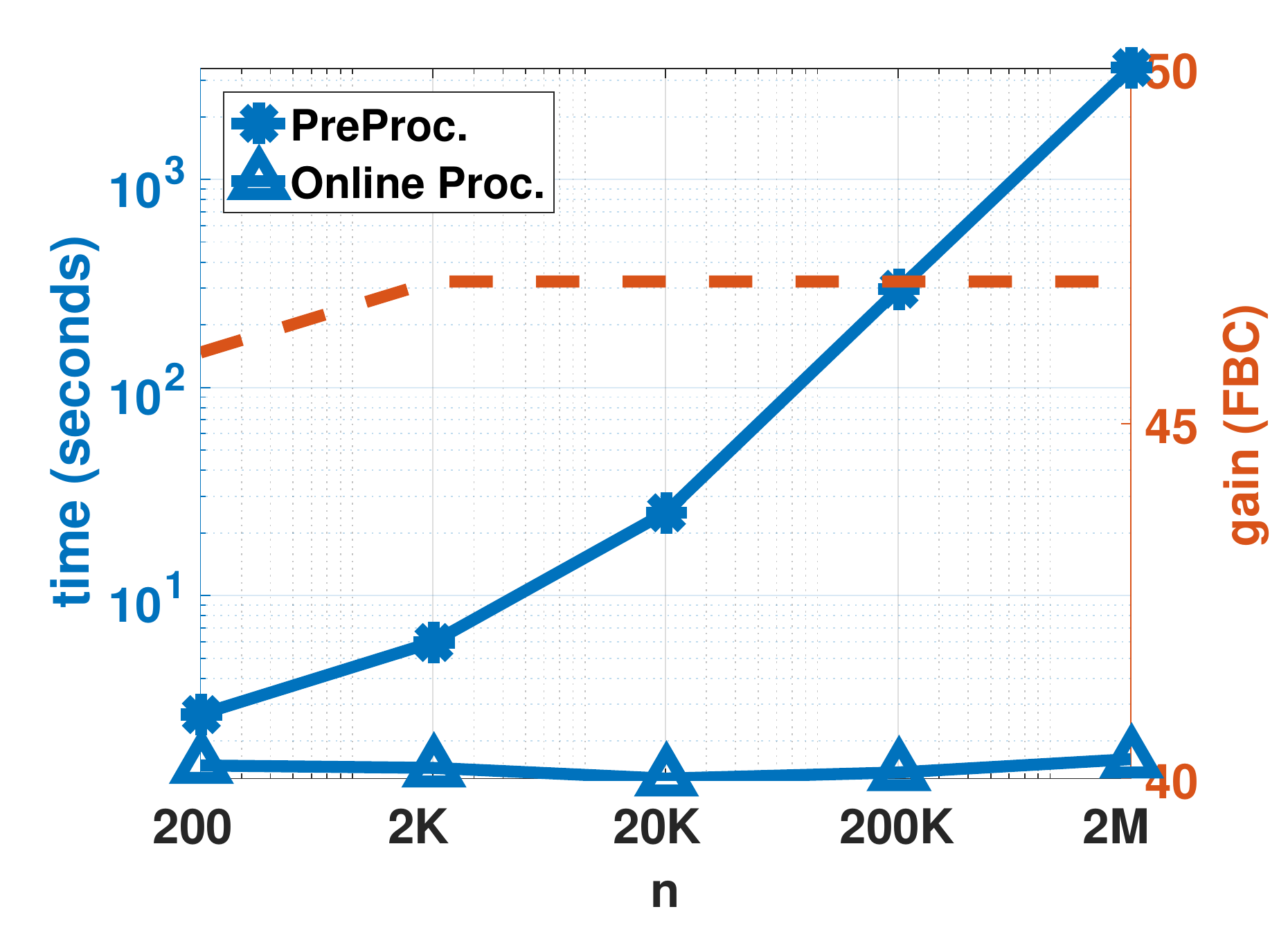}
        \caption{Impact of varying $n$ on preprocessing and online processing}
        \label{fig:vn2}
    \end{minipage}
    \hspace{1mm}
    \begin{minipage}[t]{0.3\linewidth}
       \centering
       \vspace{3mm}
        \includegraphics[scale=0.24]{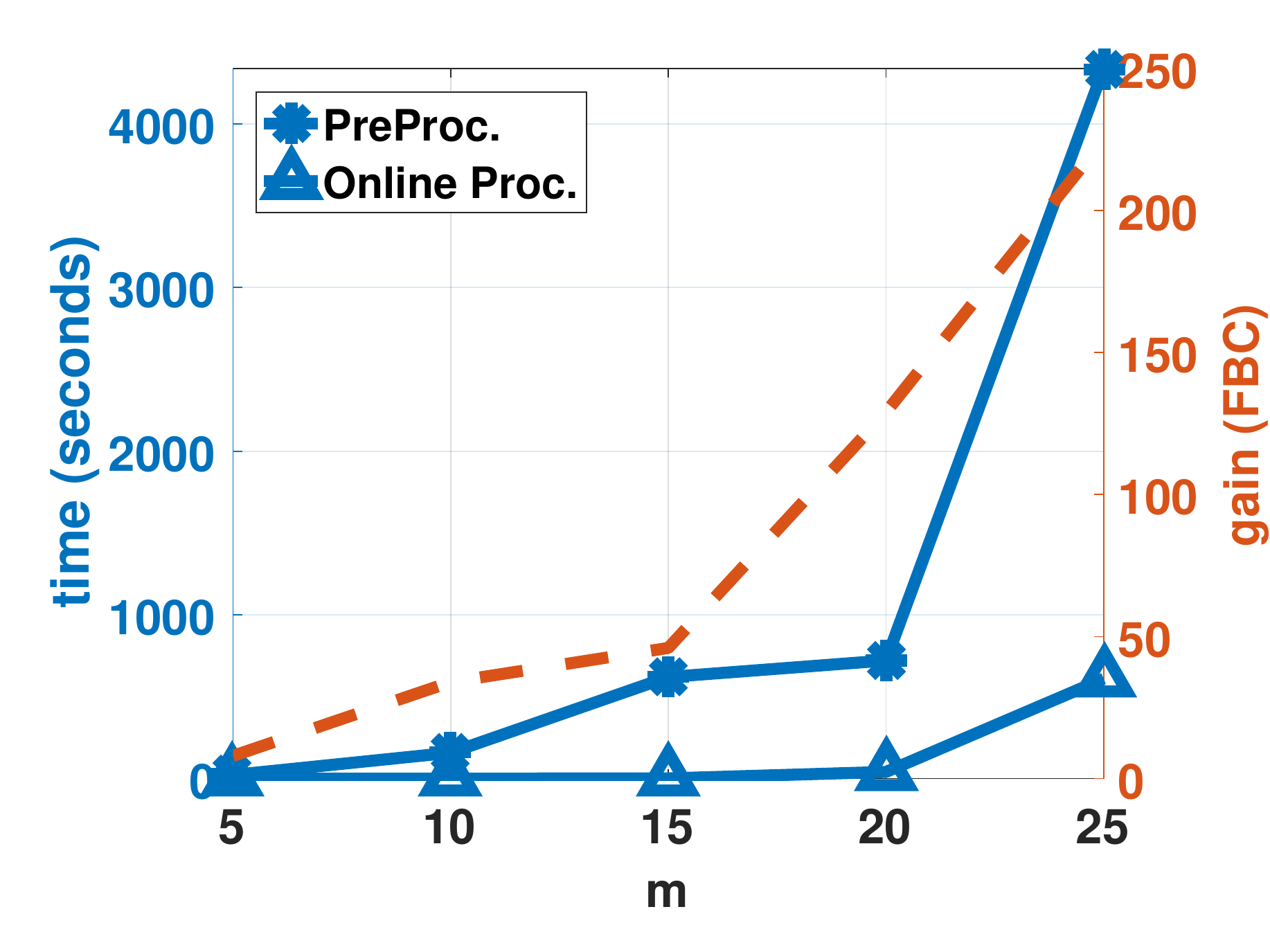}
        \caption{Impact of varying $m$ on preprocessing and online processing}
        \label{fig:vm3}
    \end{minipage}
\end{figure*}

\subsection{Experimental Results}\label{subsec:expresults}

\vspace{0.05 in}
\noindent {\bf Figure \ref{fig:vm1}: Impact of varying $m$ on GMFA algorithms.}
We first study the impact of the number of attributes ($m$) on the performance \S~\ref{sec:exact} algorithms where the GMFA problem is considered over the general 
class of monotonic gain functions. Thus, the choice of the gain function will not affect the performance of GMFA algorithms.
In this (and next) experiment we only consider the running time for the GMFA algorithms by deducting the gain function computation from the total time.
Figure~\ref{fig:vm1} presents the results for varying $m$ from $5$ to $25$ 
while keeping the other variables to their default values.
The size of $\mathcal{L}_\mathcal{A}$ exponentially depends on the number of attributes ($m$).
Thus, traversing the complete lattice, {\bf B-GMFA} did not extend beyond $15$ attributes (requiring more than $10,000$ seconds to terminate).
Utilizing the monotonicity of the gain function and pruning the nodes that are not maximal affordable, {\bf I-GMFA} could extend up to $20$ attributes. Still it requires $1566$ sec for {\bf I-GMFA} to finish with $20$ attributes and failed to extend to $25$ attributes.
Despite the exponential growth of $\mathcal{L}_\mathcal{A}$, transforming the lattice to a tree structure, reordering the attributes, and amortizing the computation costs over the nodes of the lattice, {\bf G-GMFA} performed well for all the settings requiring $116$ seconds to finish for $25$ attributes.

\vspace{0.05 in}
\noindent {\bf Figure \ref{fig:vb1}: Impact of varying $B$ on GMFA algorithms.}
In this experiment, we vary the budget from $\$ 1000$ up to $\$ 3000$ while setting other variables to their default values.
Figure~\ref{fig:vb1} presents the result of this experiment.
Since {\bf B-GMFA} traverses the complete lattice and its performance does not depend on $B$, the algorithm does not perform well and 
performance for all settings is similar.
For a small budget, the maximal affordable nodes are in the lower levels of the lattice and {\bf I-GMFA} (as well as {\bf G-GMFA}) 
require to traverse more nodes until they terminate their traversal. As the budget increases, the maximal affordable nodes move to the top of the 
lattice and the algorithm stops earlier. This is reflected in the performance of {\bf I-GMFA} in Figure~\ref{fig:vb1}.
To prevent bad performance when the budget is limited one may identify the node 
$v_i$ with the cheapest attribute combination and start the traversal from level $\ell (v_i)$. Identifying 
$\mathcal{A}_i$ can be conducted in $O(m)$ (over the sorted set of attributes) adding the cheapest attributes to $\mathcal{A}_i$ until $cost(v_i)$ is no more than $B$. No other maximal affordable node can have more attributes than $\ell (v_i)$, 
since otherwise, $v_i$ can not be the node with the cheapest attribute combination.

\vspace{0.05 in}
\noindent {\bf Figure \ref{fig:vn1}: Impact of varying $n$ on FBC algorithms.}
In next three experiments we compared the performance {\bf FBC} to that of ({\bf A-FBC}). The algorithms are utilized inside {\bf G-GMFA} 
computing the FBC (gain) of maximal frequent nodes. When comparing the performance of these algorithms, 
after running {\bf G-GMFA}, we consider the total time of each run.
The input to {\bf FBC} is the set of maximal frequent nodes, which is computed offline. Thus, during the identification of the FBC of a node $v_i$ there is no need to recompute them. However, for a fair comparison, we include in the graphs a line 
that demonstrates total time (the preprocessing time to identify maximal affordable nodes and the time required by {\bf FBC}).
We vary the number of tuples ($n$) from $200$ to $2M$ while setting the other variables to the default values.
Figure~\ref{fig:vn1} presents the results for this experiment. As reflected in the figure, {\bf A-FBC} does not extend beyond $20$K tuples 
and even for $n=20,000$ it requires more than $17,400$ seconds to complete.
Even considering preprocessing in the total running time of {\bf FBC} (blue line), it extends to $2$M tuples with a total time less than $3,500$ seconds 
(almost all the time spent during preprocessing). Note that the running time of {\bf FBC} itself does not depend on $n$. This is reflected in Figure~\ref{fig:vn1}, as in all settings (even for $n=2,000,000$) the time required by {\bf FBC} is less than $2$ seconds.

\vspace{0.05 in}
\noindent {\bf Figure \ref{fig:vm2}: Impact of varying $m$ on FBC algorithms.}
Next, we ran experiments to study the performances of {\bf A-FBC} and {\bf FBC} while varying the number of attributes from $5$ to $25$ (the other variables were set to the default values).
We utilize the two algorithms inside {\bf GMFA} and assess the time required 
by {\bf A-FBC} and {\bf FBC}; for a fair comparison, we also add a line to present the total time (preprocessing and the actual running time of {\bf FBC}).
Figure~\ref{fig:vm2} presents the results of this experiment.
{\bf A-FBC} requires $10,500$ seconds for only $10$ attributes and does not extend beyond that number of attributes in a reasonable amount of time.
On the other hand, {\bf FBC} performs well for all the settings; the total time for preprocessing and running time for {\bf FBC} on $25$ 
attributes is around $4,000$ seconds, while the time to run {\bf FBC} itself is $510$ seconds.

\vspace{0.05 in}
\noindent {\bf Figure \ref{fig:vt}: Impact of varying $\tau$ on FBC algorithms.}
We vary the frequency threshold ($\tau$) from $0.05$ to $0.2$ while keeping the other variables to their default values.
{\bf A-FBC} did not complete for any of the settings!
Thus, in Figure~\ref{fig:vt} we present the performance of {\bf FBC}; being consistent with the previous 
experiments we also present the total time (preprocessing time and the time required by {\bf FBC}).
To demonstrate the relationship between the gain and the time required by the algorithm, 
we add the FBC of the optimal solution in the right-y-axis and the dashed orange line.
When the threshold is large ($0.2$ in this experiment), for a node $v_i$ to be frequent, more tuples in $\mathcal{D}$ should have $\mathcal{A}_i$ 
in their attributes in order to be frequent. As a result, smaller number of nodes are frequent and the maximal frequent n
odes appear at the lower levels of the lattice. Thus, preprocessing stops earlier and we require less time to identify the set of maximal frequent nodes.
Also, having smaller number of nodes for larger thresholds, the gain (FBC) of the optimal solution decreases as the threshold increases. 
In this experiment, the gain of the optimal solution for $\tau=0.05$ was $129$ while it was $22$ for $\tau=0.2$.
As per Theorem 3, {\bf FBC} is output sensitive and its worst-case time complexity linearly depends on its output value. This is reflected in this experiment; the running time of {\bf FBC} drops when the FBC of the optimal solution drops. The time required by {\bf FBC} in this experiment is less than 
$3.2$ seconds for all settings.

\vspace{0.05 in}
\noindent {\bf Figures \ref{fig:vn2} and~\ref{fig:vm3}: Impacts of varying $n$ and $m$ in the online and offline running times.}
We evaluate the performance of {\bf G-GMFA} with {\bf FBC} and perform two experiments to study offline processing time (identifying the maximal affordable nodes) and the total online (query answering) time, namely the total time taken by {\bf G-GMFA} and {\bf FBC} in each run.
We also include the right-x-axis and the dashed orange line to report the gain (FBC) of the optimal solution.
First, we vary the number of tuples ($n$) between $200$ and $2$M while setting the other variables to the default values.
The results are presented in Figure~\ref{fig:vn2}. While the preprocessing time increases from less than $3$ seconds for $n=200$ up to around 
$3,400$ seconds for $n=2,000,000$, the online processing time (i.e., {\bf GMFA} and {\bf FBC} times) does not depend on $n$, and for all the settings (including $n=2,000,000$) requires less than $2$ seconds to complete.
This verifies the suitability of the proposed algorithm for online query answering for datasets at scale.
Next, we vary $m$ from $5$ to $25$ and set the other variables to default values.
While increasing $m$ exponentially increases the size of $\mathcal{L}_{\mathcal{A}}$, it also increases preprocessing time to identify
 the maximal frequent nodes. In this experiment the (offline) preprocessing takes between $22$ seconds for $m=5$ and $4,337$ seconds for $m=25$.
The increase in the lattice size increases the gain (FBC) of the optimal solution (right-y-axis)
from $8$ for $m=5$ to $225$ for $m=25$, so does the online processing time. Still the total time to execute {\bf G-GMFA} 
with {\bf FBC} for $25$ attributes is around $600$ seconds.
In practice, since $\mathcal{A}_t$ is probably not empty set, the number of remaining attributes is less than $m$ and one
may also select a subset of them as flexible attributes; we utilize $m=25$ to demonstrate that even for the extreme cases 
where the number of attributes to be considered is large, the proposed algorithms are efficient and provide query answers in reasonable time.
\subsection{Case Study}\label{subsec:casestudy}
We preformed a real case study on the actual AirBnB rental accommodations in two popular locations: {\bf Paris} and {\bf New York City}.
We used the location information of the accommodations, i.e., latitude and longitude for the filtering, and found $42,470$ rental accommodations in Paris and $37,297$ ones in New York City.
We considered two actual accommodations, one in each city, offering the same set of amenities.
These accommodations lack providing the following amenities: 
\texttt{Air Conditioning}, \texttt{Breakfast}, \texttt{Cable TV}, \texttt{Carbon Monoxide Detector}, \texttt{Doorman}, \texttt{Dryer}, \texttt{First- Aid Kit}, \texttt{Hair Dryer}, \texttt{Hot Tub}, \texttt{Indoor Fireplace}, \texttt{Internet}, \texttt{Iron}, \texttt{Laptop Friendly Workspace}, \texttt{Pool}, \texttt{TV}, and \texttt{Washer}.
We used the same cost estimation discussed in \S~\ref{subsec:exp-setup}, assumed the budget $B=\$ 2000$ for both accommodations, and ran {\bf GMFA} while considering {\bf FBC} (with threshold $\tau=0.1$) as the gain function. It took $0.195$ seconds to finish the experiment for Paris $0.365$ seconds for New York City.
While the optimal solution suggests offering \texttt{Breakfast}, \texttt{First Aid Kit}, \texttt{Internet}, and \texttt{Washer} in Paris, it suggests adding
\texttt{Carbon Monoxide Detector}, \texttt{Dryer}, \texttt{First Aid Kit}, \texttt{Hair Dryer}, \texttt{Inter-\\net}, \texttt{Iron}, \texttt{TV}, and \texttt{Washer}
in New York City.
Comparing the results for the two cases reveals the popularity of providing \texttt{Breakfast} in Paris, whereas the combination of \texttt{Carbon Monoxide Detector}, \texttt{Dryer}, \texttt{Hair Dryer}, \texttt{Iron}, and \texttt{TV} are preferred in New York City.
\section{Related Work}\label{sec:related}
\vspace{0.05in}
\noindent\textbf{Product Design}: 
The problem of product design has been studied by many disciplines such as economics, industrial engineering, operations research, and computer science~\cite{albers1980, selker2000, shocker1974}. More specifically, manufacturers want to understand the preferences of their customers and potential customers for the products and services they are offering or may consider offering. Many factors like the cost and return on investment are currently considered. Work in this domain requires direct involvement of consumers, who choose preferences from a set of existing alternative products. While the focus of existing work is to identify the set of attributes and information collection from various sources for a single product, our goal in this paper is to use the existing data for providing a tool that helps service providers as a part of peer to peer marketplace. In such marketplaces, service providers (e.g. hosts) are the customers of the website owner (e.g., AirBnB) who aim to list their service for other customers (e.g. guests) which makes the problem more challenging. The problem of item design in relation to social tagging is studied in \cite{das2011}, where the goal is to creates an opportunity for designers to build items that are likely to attract desirable tags when published. \cite{miah2009} also studied the problem of selecting the snippet for a product so that it stands out in the crowd of existing competitive products and is widely visible to the pool of potential buyers. However, none of these works has studied the problem of maximizing gain over flexible attributes.

\vspace{0.05in}
\noindent\textbf{Frequent itemsets count}: 
Finding the number of frequent itemsets and number of maximal frequent itemsets has been shown to be \#P-complete \cite{han2000,gunopulos2003discovering}. The authors in~\cite{geerts2005} provided an estimate for the number of frequent itemset candidates containing k elements rather than true frequent itemsets. Clearly, the set of candidate frequent itemsets can be much larger than the true frequent itemset. In~\cite{lhote2005average} the authors theoretically estimate the average number of frequent itemsets under the assumption that the transactions matrix is subject to either simple Bernoulli or Markovian model. In contrast, we do not make any probabilistic assumptions about the set of transactions and we focus on providing a practical exact algorithm for the frequent-item based count.

\section{Final Remarks}\label{sec:conclusion}
We proposed the problem of gain maximization over flexible attributes (GMFA) in the context of peer to peer marketplaces. Studying the complexity of the problem and the difficulty of designing an approximate algorithm for it, we provided a practically efficient algorithm for solving GMFA in a way that it is applicable for any arbitrary gain function, as long as it is monotonic.
We presented frequent-item based count (FBC), as an alternative practical gain function in the absence of extra information such as user preferences, and proposed an efficient algorithm for computing it. The results of the extensive experiment on a real dataset from AirBnB and the case study confirmed the efficiency and practicality of our proposal.

The focus of this paper is such that it works for various application specific gain functions.
Of course, which gain function performs well in practice depends on the application and the availability of the data.
A comprehensive study of the wide range of possible gain functions and their applications left as future work.
\bibliographystyle{abbrv}
\bibliography{ref}
\end{document}